\documentclass[11pt,a4paper]{article}
\usepackage[a4paper]{geometry}
\usepackage[T1]{fontenc}
\usepackage[utf8]{inputenc}
\RequirePackage[numbers]{natbib}
\usepackage{amsmath,amsfonts,amssymb,amsthm}
\usepackage{natbib}
\usepackage[usenames,dvipsnames,svgnames,table]{xcolor}
\usepackage{mathtools}
\usepackage{mathrsfs}	
\usepackage{booktabs}
\usepackage{graphicx}
\usepackage{subfigure}
\usepackage{caption}
\usepackage{enumerate}
\usepackage[normalem]{ulem}
\usepackage{float}

\mathtoolsset{showonlyrefs=true}
\usepackage{titlesec}
\numberwithin{table}{section}
\usepackage[hyperindex,breaklinks]{hyperref}
\hypersetup{frenchlinks=true}

\makeatletter
\def\timenow{\@tempcnta\time
\@tempcntb\@tempcnta
\divide\@tempcntb60
\ifnum10>\@tempcntb0\fi\number\@tempcntb
:\multiply\@tempcntb60
\advance\@tempcnta-\@tempcntb
\ifnum10>\@tempcnta0\fi\number\@tempcnta}
\makeatother

\newtheorem{lemma}{Lemma}
\newtheorem{proposition}{Proposition}

\newtheorem{definition}{Definition}
\theoremstyle{definition}
\newtheorem{assumption}{Assumption}

\newcommand{\cA}{\mathcal{A}}
\newcommand{\cB}{\mathcal{B}}
\newcommand{\cC}{\mathcal{C}}
\newcommand{\cD}{\mathcal{D}}

\newcommand{\cF}{\mathcal{F}}
\newcommand{\cH}{\mathcal{H}}
\newcommand{\cL}{\mathcal{L}}
\newcommand{\cM}{\mathcal{M}}
\newcommand{\cN}{\mathcal{N}}

\newcommand{\cR}{\mathcal{R}}
\newcommand{\cS}{\mathcal{S}}

\newcommand{\cV}{\mathcal{V}}

\newcommand{\crO}{\mathscr{O}}

\newcommand{\E}{\mathbb{E}}
\newcommand{\F}{\mathbb{F}}

\renewcommand{\P}{\mathbb{P}}

\newcommand{\R}{\mathbb{R}}

\newcommand{\dd}{\mathrm{d}}
\newcommand{\ee}{\mathrm{e}}

\newcommand{\del}{\partial}

\renewcommand{\epsilon}{\varepsilon}

\newcommand{\tilO}{\tilde{\mathcal{O}}}

\newcommand{\eqlnostar}[2]{\begin{align}\label{#1}#2\end{align}}
\newcommand{\eqstar}[1]{\begin{align*}#1\end{align*}}

\title{Portfolio Benchmarking under Drawdown Constraint and Stochastic Sharpe Ratio}

\author{Ankush Agarwal\thanks{Centre de Math\'ematiques Appliqu\'ees, \'Ecole Polytechnique and CNRS, Route de Saclay, 91128 Palaiseau Cedex, France; Email: {\tt ankush.agarwal@polytechnique.edu}. The author research is  part of the Chair {\it Financial Risks} of the {\it Risk Foundation}.} \and Ronnie Sircar\thanks{ORFE Department, Sherrerd Hall, Princeton University,  Princeton NJ 08544; Email: {\tt sircar@princeton.edu}. The author research is partially supported by NSF grant DMS-1211906.}}

\date{This version: \today} 

\begin{document}

\maketitle

\begin{abstract}
We consider an investor who seeks to maximize her expected utility derived from her terminal wealth relative to the maximum performance achieved over a fixed time horizon, and under a portfolio \emph{drawdown} constraint, in a market with local stochastic volatility (LSV). In the absence of closed-form formulas for the value function and optimal portfolio strategy, we obtain approximations for these quantities through the use of a coefficient expansion technique and nonlinear transformations. We utilize regularity properties of the \emph{risk tolerance} function to numerically compute the estimates for our approximations. In order to achieve similar value functions, we illustrate that, compared to a constant volatility model, the investor must deploy a quite different portfolio strategy which depends on the current level of volatility in the stochastic volatility model.
 
\vspace{0.5pc}

\vspace{0.5pc}


\vspace{0.5pc}


\end{abstract}

\section{Introduction}
\label{sec:intro}
\subsection{Background and motivation}
In the vast and long-dated literature on dynamic portfolio optimization, different types of terminal utility paradigms under various portfolio constraints have been considered to understand investor behaviour (see, for instance, \citet{rogers2013optimal} for a detailed exposition). The solutions to these problems provide optimal investment strategies which aid institutional investors, and at times help to reveal deep insights about market observed phenomenons. The classical problem of continuous-time portfolio optimization dates back to \citet{samuelson1969lifetime} and \citet{merton1971optimum,merton1969lifetime}. In his seminal paper, \citet{merton1971optimum} considered a market where the prices of risky assets are given by geometric Brownian motions (with \emph{constant} volatilities), and the objective is to maximize the expected utility of terminal wealth by investing capital between the risky assets and a risk-free bank account. For constant relative risk aversion utility (CRRA) functions, the author showed that the optimal strategy is a ``fixed mix'' investment in the risky assets and the bank account. 

Merton's landmark result provided structural market insight but the restrictive problem setting -- investor objective and market dynamics -- prevented application of the results to practical situations. As a result, subsequent research has focused upon relaxing the assumptions made in \cite{merton1971optimum}, incorporating various market constraints and considering more realistic model settings. 

Portfolio managers typically use a \emph{stop-loss} level on the portfolio value to prevent a complete wipe-out of  wealth in the face of falling prices. This is also known more commonly as the \emph{drawdown} constraint. Under this constraint, the wealth in the portfolio must always remain above a certain fraction of the current maximum wealth value achieved. Furthermore, in several instances, portfolio managers commit a certain percentage of the starting wealth to the pooling investors. This situation is also covered by imposing a drawdown constraint on the portfolio wealth.

In this article, we propose a new framework to study the dynamic portfolio optimization under a drawdown portfolio constraint in a stochastic volatility market model. In many empirical studies it has been well established that stochastic volatility is a reasonable asset price modelling tool to capture the market observed volatility smiles and volatility clustering. Our principal innovation is to introduce a new terminal investor objective paradigm which allows for a reduction in the dimensionality of the problem. As our central objective in this work is to numerically study the impact of stochastic volatility on the value function and optimal portfolio strategy, the dimensionality reduction serves as a crucial feature to allow for an efficient implementation of the numerical procedures used to solve the problem and study the effects of stochastic volatility. 
\subsection{Literature review}
Several authors have considered the optimal portfolio problems under drawdown constraint. \citet{grossman1993optimal} were the first to comprehensively study this problem over infinite time horizon in a lognormal market model. They investigated to maximize the long term growth rate of the expected utility of the wealth and used dynamic programming principle to solve the problem. \citet{cvitanic1995portfolio} streamlined the analysis of \citet{grossman1993optimal} and extended the results to the case when there are multiple risky assets whose dynamics are governed by a lognormal model with deterministic coefficients. By defining an auxiliary process, they were able to show that the solution of optimization problem with drawdown constraint can be linked to an unconstrained optimization problem whose solution follows from the work of \citet{karatzas1987optimal}. They further showed that in the case of logarithmic utility function, the results hold even if the coefficients in the lognormal model are random and satisfy some ergodicity condition. In \cite{sekine2013long}, Sekine carried forward the arguments and results of \citet{cvitanic1995portfolio} to a multi-asset market model with single stochastic volatility factor. More recently, Cherny and Ob{\l}{\'o}j \cite{cherny2013portfolio} have studied the optimal portfolio problem in an abstract semimartingale model with a generalized drawdown constraint. They utilized the properties of Az\'{e}ma-Yor processes to show that the value function of the constrained problem, where the investor objective is to maximize the long term growth rate of the expected utility, has the same value function as an unconstrained problem with a suitably modified utility function. Moreover, they showed that the optimal wealth process can also be obtained as an explicit pathwise transformation of the optimal wealth process in the unconstrained problem. 

The portfolio optimization problem with drawdown constraint has also been studied in a continuous-time framework with consumption. \citet{roche2006optimal} studied the problem of maximizing the expected utility of consumption over an infinite time horizon for a power utility function under a linear drawdown constraint. This analysis was performed in the setting of a lognormal model with single asset. \citet{elie2008optimal} subsequently generalized the result to a general class of utility functions in the setting of zero interest rates and obtained an explicit representation of the solution. \citet{elie2008finite} also studied a finite time version of the same problem and in the absence of analytical representation, he provided a numerical solution to the problem.

In the financial literature, different problem settings with a drawdown constraint have received considerable attention due to their significance. \citet{magdon2004maximum} considered the problem of optimal portfolio choice when the drawdown is minimized in the single asset lognormal market model. \citet{chekhlov2005drawdown} analyzed the portfolio optimization problem in discrete time where the investor objective is to maximize the expected return from the portfolio subject to risk constraints given in terms of drawdowns. They considered a multi-asset market model and reduced the problem to a linear programming problem which can be solved numerically. In the insurance literature, drawdown constraint has been incorporated to study problems of lifetime investments. In \cite{chen2015minimizing}, Chen et al.\ considered the optimization problem of minimizing the probability of a significant drawdown occurring over a lifetime investment, i.e.\ the probability that portfolio wealth hits the drawdown barrier before a random time which represents the death time of a client. 

\subsection{Our contributions}
In this article, we consider an investor who at any time is worried about her wealth falling below a fixed fraction of the running maximum wealth and, thus, is only interested to maximize the ratio of these two quantities at the end of a fixed investment horizon. As the investor is cautious about the drawdown, consequently it is not possible to achieve an unreasonable amount of wealth by looking at an unbounded terminal utility. Therefore, it is sensible to consider a bounded terminal utility. The proposed investor objective paradigm is also motivated from the perspective of portfolio benchmarking and fixed target problems. In our setting, we start from an initial value of the maximum wealth which satisfies the drawdown constraint. The portfolio strategy in our problem allows the portfolio wealth to hit the level of initial maximum wealth by investing in the risky asset thus hitting the target or benchmark. Heuristically, it can also be deduced that the optimal portfolio strategy will liquidate the position in the risky asset once the maximum wealth is reached. This mimics the logic of classical Merton strategy which suggests to sell the risky asset close to the highest value of the portfolio.

We consider the basic setting of a frictionless financial market with a single underlying asset and a risk-free money market account. We study this problem in a stochastic volatility environment to demonstrate how uncertainty in the volatility impacts the optimal portfolio strategy. This problem has no explicit solution and thus, we look for accurate approximations to the value function and optimal strategy. We use the technique of coefficient expansion to formulate separate problems for different terms in the expansion of value function. The solutions to these problems allow us to derive an expansion for the optimal portfolio strategy. Due to the presence of portfolio constraints, the expansion terms in the value function approximation are not available in closed-form. We numerically solve for the leading term in the value function approximation and use the regularity properties of the so-called \emph{risk tolerance} function to compute the remaining higher order expansion terms. The numerical estimates for the optimal portfolio strategy are derived similarly. 

We show that the leading terms in the expansion of value function and optimal strategy are related to the solution of our problem in a lognormal asset pricing model with constant volatility. The optimal strategy in this case suggests to liquidate the risky position when portfolio wealth approaches its maximum value. Also, close to the drawdown constraint, the optimal strategy instructs to steadily build up a position in the risky asset to drive away the portfolio value from the lower barrier. The stochastic volatility correction term for the value function suggests very small loss or gain due to the uncertainty in volatility. However, we observe that depending on the current level of stochastic volatility, the optimal strategy with volatility correction is remarkably different than the case with constant volatility. Close to the maximum wealth value, the corrected optimal strategy suggests to hold onto the risky assets longer than in the constant volatility case. This clearly illustrates the impact of stochastic volatility on the optimal investment strategy. However, near the drawdown barrier, the behavior of corrected optimal strategy depends on the level of current stochastic volatility in the model when compared to the optimal strategy in the constant volatility case. 

\subsection{Organization}
In Section \ref{sec:problem formulation} we introduce the continuous-time model setting and formulate the problem. We derive the HJB equation for the optimal portfolio problem and give the analytical formula for the optimal portfolio strategy in terms of the value function. We provide the approximation formulas for the value function and optimal portfolio strategy in Section \ref{sec:main results} and summarize our main results. In Section \ref{sec:numerics}, we discuss the numerical implementation of our results and provide practical insights with the help of popular numerical examples considered in the literature. Section \ref{sec:conclusion} concludes the article and suggests directions for future research. The proofs are included in Appendix \ref{sec:proofs}.

\section{Problem Formulation}
\label{sec:problem formulation}
We consider a complete filtered probability space $(\Omega,\cF, \{\cF_t\}_{t \geq 0}, \P)$ endowed with a two dimensional Brownian motion $W = \bigl((W^{(1)}_t,W^{(2)}_t), 0 \leq t \leq T \bigr)$ and suppose there is a risky asset whose dynamics under $\P$ is given by the following local stochastic volatility (LSV) model: 
\eqstar{
\frac{\dd S_t}{S_t} &= \tilde{\mu}(S_t,Y_t)\dd t + \tilde{\sigma}(S_t,Y_t) \dd B^{(1)}_t,\\
\dd Y_t &= c(Y_t) \dd t + \beta(Y_t) \dd B^{(2)}_t,
}
where $B^{(1)}_t := W^{(1)}_t $ and $B^{(2)}_t := \rho W^{(1)}_t + \sqrt{1 - \rho^2} W^{(2)}_t$ are standard Brownian motions under measure $\P$ with correlation $\rho \in [-1,1]: \dd\langle B^{(1)}_t B^{(2)}_t \rangle = \rho \dd t.$ From It\^{o}'s formula, the log price process $X = \log S$ is described as following: 
\eqstar{
\dd X_t &= b(X_t,Y_t) \dd t + \sigma(X_t,Y_t)\dd B^{(1)}_t,
}
where $\mu(X_t,Y_t) :=\tilde{\mu}(\ee^{X_t},Y_t), \sigma(X_t,Y_t) := \tilde{\sigma}(\ee^{X_t},Y_t)$ and 
\eqstar{
b(X_t,Y_t) := \mu(X_t,Y_t) - \frac{1}{2} \sigma^2(X_t,Y_t).
} 
We assume that the model coefficient functions $\mu, \sigma, c$ and $\beta$ are Borel-measurable and possess sufficient regularity to ensure that a unique strong solution exists for $(X,Y)$ which is adapted to the augmentation $\mathbb{F} = \{\cF_t: 0 \leq t \leq T \}$ of the filtration generated by $W.$ 

Further, we suppose the existence of a frictionless financial market with the price of a single risky asset given by $S$ and the risk-free rate of interest given by a scalar constant $r > 0.$ In this market, we denote the wealth process of an investor by $\bar{L}$ who invests $\bar{\pi}_t$ units of currency in risky asset $S$ at time $t$ and the remaining $(\bar{L}_t - \bar{\pi}_t)$ units of currency in the risk-free bank account. Then, the self-financing portfolio, $\bar{L}$ satisfies the following stochastic differential equation (SDE)
\eqstar{
\dd\bar{L}_t &= r (\bar{L}_t - \bar{\pi}_t) \dd t + \bar{\pi}_t \frac{\dd S_t}{S_t}\\
&= \left(r \bar{L}_t + \bar{\pi}_t (\mu(X_t,Y_t) - r)\right) \dd t + \bar{\pi}_t \sigma(X_t,Y_t)\, \dd B^{(1)}_t.
}

The running maximum wealth in time $t$ dollars is given by $\bar{M}_t := \max\{\bar{L}_s \ee^{r(t-s)};s \leq t\}.$ In this work, we propose an investment framework that encourages exiting the market in the face of a sizable drawdown, while also targeting a benchmark that is related to the running maximum, or high watermark of the investment performance. The investor's risk preferences are given by a utility function $U$ satisfying: 
\assumption{
\label{ass:utilfunc}
The terminal utility function $U: (\alpha,1) \to \R,$ is smooth: $U \in \cC^\infty(\alpha,1).$ It is also strictly increasing and strictly concave.
}

We solve the utility maximization problem at finite $T > 0$  with the \emph{drawdown constraint}: 

$$\bar{L}_t \geq \alpha \bar{M}_t \text{ a.s., }\quad 0 \leq t \leq T,\quad \mbox{where $\alpha \in (0,1)$ is a fixed drawdown parameter.}$$

\subsection{The discounted formulation}
\label{sec:discounted formulation}
We look to formulate the problem in the setting where the wealth process is discounted with respect to the risk-free rate of interest. This allows us to clearly study the impact of stochastic volatility on the optimal strategy and value function. For this purpose, we define, $L_t := \bar{L}_t \ee^{-rt}$ and $M_t := \bar{M}_t\ee^{-rt} = \max\{L_s; s\leq t\}.$ The discounted wealth process satisfies the following SDE
\eqstar{
\dd L_t &= \pi_t \bigl( (\mu(X_t,Y_t)- r) \dd t + \sigma(X_t, Y_t) \dd B^{(1)}_t \bigr),
}
where $\pi_t := \ee^{-rt}\bar{\pi}_t$ is the risky-asset trading strategy. 

Now, we are ready to express the investor's utility maximization problem by defining the value function 
\eqlnostar{eq:mainprobBStrans}{
V(t,l,m,x,y) = \sup_{\pi \in \Pi} \E \left[U\left( \frac{L_T}{M_T}\right) \Big\vert L_t = l, M_t = m, X_t = x, Y_t = y \right], 
}
 where the admissible strategies are given by
\eqstar{
\Pi_{\alpha,t,l,m} := &\bigl\{ \pi: \text{measurable }, \F-\text{adapted}, \E_{t,l,m,x,y}\int^T_t \pi^2_s \sigma^2(X_s,Y_s)\dd s < \infty,\\
& \text{ s.t. }  L_s \geq \alpha M_s > 0\text{ a.s.}, t \leq s \leq T \bigr\}.
} 
We define the domain in  $\R_+ \times \R^4$ as $[0,T) \times \tilO_{\alpha}$ where $$\tilO_{\alpha}:= \{(l,m,x,y): 0 < \alpha m < l < m\}.$$  
The above value function $V$ is defined for any $5-$tuple $(t,l,m,x,y) \in [0,T] \times \overline{\tilO}_{\alpha}$. 

We recall that $\dd M = 0$ on $\{t \geq 0 | M_t \neq L_t \}.$ Then, for $(t,l,m,x,y) \in [0,T) \times \tilO_{\alpha}$ and $V \in C^{1,2,1,2,2} ([0,T] \times \overline{\tilO}_{\alpha}),$ following the usual dynamic programming principle (see, for example, \citet[Chapter 3 ]{pham2009continuous}), we obtain the Hamilton-Jacobi-Bellman (HJB) equation
\eqstar{
(\del_t + \cA) V + \sup_{\pi \in \R} \cA^{\pi} V = 0,
}
where $(\cA + \cA^{\pi})$ is the generator of the process $(X,Y,L)$ with
\eqstar{
\cA &= b(x,y)\frac{\del}{\del x} + c(y) \frac{\del}{\del y} + \frac{1}{2} \sigma^2(x,y)\frac{\del^2}{\del x^2} + \frac{1}{2}\beta^2(y) \frac{\del^2}{\del y^2} + \sigma(x,y)\beta(y)\rho \frac{\del^2}{\del x \del y},\\
\cA^{\pi} &= \pi \Bigl[(\mu(x,y)-r) \frac{\del}{\del l}  + \sigma^2(x,y) \frac{\del^2}{\del x\del l} + \rho\sigma(x,y)\beta(y)\frac{\del^2}{\del y \del l}\Bigr]+ \frac{1}{2}\pi^2 \sigma^2(x,y) \frac{\del^2}{\del l^2}.
}
By inspecting the quadratic expression above in $\pi,$ it is clear that the optimal strategy $\pi^* := \underset{{\pi \in \R}}{\arg \max} \cA^{\pi} V$ is given as 
\eqlnostar{eq:optportfolio}{
\pi^*= -\frac{(\mu(x,y)-r) V_l + \rho \beta(y) \sigma(x,y) V_{yl} + \sigma^2(x,y) V_{xl}}{\sigma^2(x,y) V_{ll}},
}
where the subscripts indicate partial derivatives. The HJB equation becomes
\eqlnostar{eq:hjbmain1}{
(\del_t + \cA) V + \tilde{\cN}(V) = 0,
}
with the nonlinear term given as
\eqstar{
\tilde{\cN}(V) &= -\frac{1}{2V_{ll}}\Bigl(\lambda(x,y) V_l + \sigma(x,y) V_{xl} + \rho\beta(y) V_{yl}\Bigr)^2,
}
where 
\eqlnostar{eq:definition sharpe ratio}{
\lambda(x,y) := \frac{\mu(x,y) - r}{\sigma(x,y)}
} 
is the Sharpe ratio function. The boundary conditions are 
\eqlnostar{eq:hjbbound1}{
\text{(Terminal condition): }  \qquad &V(T,l,m,x,y) = U\left( \frac{l}{m}\right), \qquad \\
\label{eq:hjbbound2}
\text{(Neumann condition): }   \qquad &V_m(t,m,m,x,y) = 0, \qquad\\
\label{eq:hjbbound newportfolio}
\text{(Drawdown Dirichlet condition): } \qquad &V(t,\alpha m, m, x, y) = U(\alpha).
}
The above Dirichlet condition signifies that when the drawdown constraint is hit, the investor stops trading in the risky asset $(\pi_t = 0).$ In the discounted formulation when the investor stops trading, it signifies that the wealth process stops varying and the investor accepts the utility which is given at the drawdown barrier.

\subsection{Dimensionality reduction}
\label{sec:dimension reduction}
The nonlinear PDE in \eqref{eq:hjbmain1} with boundary conditions \eqref{eq:hjbbound1}, \eqref{eq:hjbbound2} and \eqref{eq:hjbbound newportfolio} is difficult to solve numerically because the domain $\tilO_\alpha$ is a wedge in $(L,M)$ space requiring a non-rectangular finite-difference grid. However, we notice that given the structure of our problem, we could perform a change of variable which reduces the dimensionality of the problem. We introduce 
\eqstar{
\xi = \frac{l}{m}, \quad \text{ and define } \quad Q(t,\xi,x,y) := V(t,l,m,x,y),
}
which results in a new non-linear PDE for $Q \in C^{1,2,2,2}\bigl([0,T] \times [\alpha,1] \times \R^2\bigr):$
\eqlnostar{eq:hjbmain2}{
(\del_{t} + \cA) Q + \cN(Q) = 0, \text{ on }  [0,T) \times (\alpha ,1) \times \R^2,
}
where
\eqstar{
\cN(Q) = -\frac{1}{2Q_{\xi \xi}}\Bigl(\lambda(x,y) Q_{\xi} + \sigma(x,y) Q_{x \xi} + \rho\beta(y) Q_{y\xi}\Bigr)^2,
}
and the boundary conditions are
\eqlnostar{fullbcns}{
Q(T,\xi,x,y) = U\left( \xi \right), \quad Q_\xi(t,1,x,y) = 0, \quad Q(t,\alpha,x,y) = U(\alpha). 
}
Apart from providing a reduction in dimensionality, the above change of variable also transforms the problem domain from a high-dimensional cone to a semi-rectangular domain which typically helps to get more accurate numerical estimates for the solution.

\section{Value Function and Optimal Strategy Approximation}
\label{sec:main results}
Even under the constant volatility lognormal asset model, no closed form solution is available for the nonlinear PDE \eqref{eq:hjbmain2} and one needs to rely on accurate numerical approximations. In this paper, we propose to find an approximation for the value function as 
\eqlnostar{eq:approx solution}{
Q = Q^{(0)} + Q^{(1)} + Q^{(2)} + \ldots, 
}
as well as an approximation for the optimal investment strategy 
\eqlnostar{eq:approx strategy}{
\pi^* = \pi_0 + \pi_1 + \pi_2 + \ldots,
}
by using the coefficient expansion technique. This approach has been developed for the linear European option pricing problem in a general LSV model setting by \citet{lorig2015explicit}, and for the classical (unconstrained) Merton problem by \citet{lorig2015portfolio}.

\subsection{Coefficient polynomial expansions}
The main idea of the coefficient expansion technique is to first fix a point $(\bar{x},\bar{y})\in \R^2$ and then for any function $\chi(x,y),$ which is locally analytic around $(\bar{x},\bar{y}),$ define the following family of functions indexed by $a \in [0,1]:$
\eqstar{
\chi^a(x,y) := \sum^{\infty}_{n=0}a^n \chi_n(x,y)
}
where 
\eqstar{
\chi_n(x,y) := \sum^n_{k=0} \chi_{n-k,k}(x-\bar{x})^{n-k} (y - \bar{y})^k, \quad \chi_{n-k,k} := \frac{1}{(n-k)!k!}\frac{\del^{n-k}}{\del x^{n-k}}\frac{\del^k}{\del y^k}\chi(x,y)\Big \vert_{x=\bar{x},y=\bar{y}}.
}
Note that for $n=0$, $\chi_0 := \chi_{0,0} = \chi(\bar{x},\bar{y})$ is a constant. We can observe that $\chi^a \Big \vert_{a=1}$ is the Taylor series expansion of $\chi$ about the point $(\bar{x},\bar{y}).$ Here, $a$ is seen as a \emph{perturbation parameter} which is used to identify the successive terms in the approximation. 

To apply this technique in PDE \eqref{eq:hjbmain2}, we first replace each of the coefficient functions 
\[ \chi\in\{b,c,\sigma^2,\beta^2,\sigma\beta,\lambda,\sigma,\beta\} \]
with their respective series expansion for some $a \in (0,1)$ and $(\bar{x},\bar{y}) \in \R^2$. Next, to obtain approximations as in \eqref{eq:approx solution} and \eqref{eq:approx strategy}, we define a series expansion of value function as $Q = Q^a =\sum^{\infty}_{n=0} a^n Q^{(n)},$ linear operator $\cA = \cA^a =  \sum^{\infty}_{n=0} a^n \cA_n$ and replace the non-linear operator $\cN(Q)$ by  $\cN^a(Q^a)$ which involves series expansions for the coefficient functions and the value function. Then from \eqref{eq:hjbmain2}, we consider the PDE problem 

\eqlnostar{eq:hjbmain3}{
(\del_t + \cA^a) Q^a + \cN^a(Q^a) = 0, \text{ on }  [0,T) \times (\alpha ,1) \times \R^2,
}
with the boundary conditions 
\eqlnostar{eq:hjbbound8}{
Q^a(T,\xi,x,y) = U\left( \xi \right), \quad Q^a_\xi(t,1,x,y) = 0, \quad
Q^a(t,\alpha,x,y) = U(\alpha).
}
Now, to obtain the successive terms of approximation in expansions \eqref{eq:approx solution} and \eqref{eq:approx strategy}, we compare the corresponding degree terms in the polynomial of perturbation parameter $a$ in \eqref{eq:hjbmain3} and the boundary conditions \eqref{eq:hjbbound8}. The approximations are then obtained by setting $a=1.$

\subsection{Zeroth and first order approximation}
\label{sec:zeroth and first}
The first term in the approximation \eqref{eq:approx solution} is obtained by collecting the zeroth order terms w.r.t. $a$ in the expansion of \eqref{eq:hjbmain3}. We get
\eqlnostar{eq:hjbzeroth}{
(\del_t + \cA_0) Q^{(0)} -\frac{1}{2Q^{(0)}_{\xi \xi}}\Bigl(\lambda_0 Q^{(0)}_{\xi} + \rho\beta_0 Q^{(0)}_{y\xi}\Bigr)^2 = 0,
}
with
\eqlnostar{eq:defA0}{
\cA_0 := b_0\frac{\del}{\del x} + c_0 \frac{\del}{\del y} + \frac{1}{2} \sigma^2_0\frac{\del^2}{\del x^2} + \frac{1}{2}\beta^2_0 \frac{\del^2}{\del y^2} + \rho \sigma_0\beta_0 \frac{\del^2}{\del x \del y},
}
and the corresponding order boundary conditions are 
\eqlnostar{eq:hjbzeroth_bound1}{
Q^{(0)}(T,\xi,x,y) = U( \xi ),\quad
Q^{(0)}_\xi(t,1,x,y) = 0,\quad
Q^{(0)}(t,\alpha,x,y) = U(\alpha).
}
As the linear operator $\cA_0$ has only constant coefficients and the boundary conditions do not depend on $(x,y)$, the solution $Q^{(0)}(t,\xi,x,y)$ is independent of $(x,y).$  Therefore, in this case we get:
\begin{definition}
\label{thm:zeroth order} 
The leading order term $Q^{(0)} = Q^{(0)}(t,\xi)$ satisfies the following nonlinear PDE
\eqlnostar{eq:hjbzeroth1}{
Q^{(0)}_t  -\frac{1}{2}\lambda_0^2 \frac{\bigl(Q^{(0)}_{\xi}\bigr)^2}{Q^{(0)}_{\xi \xi}}  = 0, \text{ on } [0,T) \times (\alpha,1),
}
with the boundary conditions 
\eqlnostar{eq:hjbzeroth_bound3}{
Q^{(0)}(T,\xi) = U( \xi ),\quad
Q^{(0)}(t,\alpha) = U(\alpha), \quad
Q^{(0)}_\xi(t,1) = 0.
}
\end{definition}
It can be seen (and also shown later) that the zeroth order term $Q^{(0)}$ actually corresponds to the value function of our investor problem which arises in the case of a constant volatility and growth rate lognormal asset price market model, with constant Sharpe ratio $\lambda_0$. Due to the presence of boundary conditions, an explicit formula for $Q^{(0)}$ is inaccessible, even for a power utility function, and we estimate the quantity through numerical techniques. This is explained in detail in Section \ref{sec:numerics}. 

\begin{assumption}\label{reg-assump} 
We assume throughout that the PDE problem \eqref{eq:hjbzeroth1}-\eqref{eq:hjbzeroth_bound3} has a unique classical solution $Q^{(0)} \in C_b^{1,5}([0,T) \times [\alpha,1])$,
that is $Q^{(0)}$ has at least five derivatives in $\xi$ which are continuous and bounded up to the boundaries at $\xi=\alpha,1$.
\end{assumption}
In the unconstrained case, with no drawdown restrictions, the PDE \eqref{eq:hjbzeroth1} is simply the constant Sharpe ratio Merton value function PDE on the half-space $\xi>0$, where $\xi$ would denote the wealth level. As is well-known, given a smooth and strictly concave utility function satisfying the usual conditions ($U'(0^+)=\infty$ and $U'(\infty)=0$), smoothness of the value function follows from Legendre transform to a linear parabolic PDE. In our restricted drawdown problem we assume regularity of the solution when restricted to a finite domain. Our value function approximation, summarized in Section \ref{summary}, and our optimal portfolio approximation in Section \ref{optport}, are given in terms of (up to $5$th order) partial derivatives of $Q^{(0)}$.

In order to find the first order correction term, we introduce the following \emph{risk tolerance} function 
\eqlnostar{eq:define risk tolerance}{
\cR(t,\xi) := \left(-\frac{Q^{(0)}_\xi}{Q^{(0)}_{\xi\xi}}\right)(t,\xi).
}
This function has been well studied in the unconstrained case by \citet{kallblad2014qualitative} and has been recently used to study the classical Merton problem in a stochastic volatility environment by \citet{fouque2015portfolio}. It satisfies an autonomous PDE of fast-diffusion type:
\begin{proposition}
\label{thm:risk tolerance}
The risk tolerance function $\cR(t,\xi)$ satisfies the nonlinear PDE
\eqlnostar{eq:risk tol pde}{
\cR_t + \frac{1}{2}\lambda_0^2 \cR^2 \cR_{\xi \xi} = 0, \text{ on } [0,T) \times (\alpha,1),
}
with the boundary conditions
\eqlnostar{eq:R_bound1}{
\cR(T,\xi) = -\frac{U^\prime(\xi)}{U^{\prime\prime}(\xi)},\quad
\cR(t,\alpha) = 0,\quad
\cR(t,1) = 0.
}
\end{proposition}
The proof is given in Appendix \ref{Prop2proof}.

As we show later in Section \ref{sec:optimal strategy}, Proposition \ref{thm:risk tolerance} is also crucial to compute the leading order terms in the approximation of optimal strategy $\pi^*.$ Next, we define the differential operators 
\eqlnostar{eq:def diff operator}{
\cD_k := \cR^k \frac{\del^k}{\del \xi^k}, \quad k = 1,2,\ldots, 
}
which allows us to write equation \eqref{eq:hjbzeroth1} as 
\eqlnostar{eq:hjbzeroth transform1}{
\Bigl( \partial_t +  \frac{\lambda_0^2}{2} \cD_2 + \lambda_0^2 \cD_1\Bigr) Q^{(0)} = 0.
}
Next, we collect the first order terms w.r.t. $a$ in the expansion \eqref{eq:hjbmain3}. As $Q^{(0)}$ does not depend on $y$, the linear term contributes
\eqstar{
\left(\frac{\del}{\del t} + \mathcal{A}_0\right)Q^{(1)},
}
and the nonlinear term contributes 
\eqstar{
\lambda^2_0 \cD_1 Q^{(1)} + \frac{1}{2} \lambda^2_0\cD_2 Q^{(1)}  + \lambda_0\lambda_1 \cD_1 Q^{(0)} + \beta_0 \lambda_0 \rho \cD_1 \frac{\del}{\del y} Q^{(1)} + \sigma_0 \lambda_0 \cD_1 \frac{\del}{\del x} Q^{(1)}.
} 
\begin{definition}
\label{def:Q1}
The first order correction term $Q^{(1)}$ satisfies the following PDE 
\eqlnostar{eq:nonlinear first order}{
\left(\frac{\del}{\del t} + \cA_0 + \cB_0 \right)Q^{(1)} + S_1 = 0, \text{ on } [0,T) \times (\alpha, 1) \times \R^2,
}
with linear operator $\cB_0$ given as 
\eqlnostar{eq:defB0}{
\cB_0 := \lambda^2_0 \cD_1 + \frac{1}{2} \lambda^2_0\cD_2 + \beta_0 \lambda_0 \rho \cD_1 \frac{\del}{\del y} + \sigma_0 \lambda_0 \cD_1 \frac{\del}{\del x},
}
and the source term
\eqstar{
S_1 = \bigl(\frac{1}{2}\lambda^2\bigr)_1(x,y) \cD_1 Q^{(0)}(t,\xi).
}
The terminal and boundary conditions \eqref{eq:hjbbound8} for $Q^a$ are already satisfied by $Q^{(0)}$, and so we have
\eqlnostar{eq:hjbbound11}{
Q^{(1)}(T,\xi,x, y) = 0, \quad
Q^{(1)}_\xi(t,1,x, y) = 0, \quad
Q^{(1)}(t,\alpha, x,y) = 0.
}
\end{definition}
In Section \ref{sec:const coefficient transform}, we show that  $Q^{(1)}$ can be expressed in terms of partial derivatives of $Q^{(0)}$ and $\cR.$

\subsubsection{Explicit expression for first order correction term}
\label{sec:const coefficient transform}
We now employ a transformation that enables us to find an explicit expression for $Q^{(1)}$ in terms of partial derivatives of $Q^{(0)}$. For this purpose, we first note that $Q^{(0)}_{\xi}$ is a monotone function from the following result on the zeroth order term. 
\begin{lemma}
\label{lem:zeroth order properties}
$Q^{(0)}(t,\xi)$ is a non-decreasing and concave function in the $\xi$ variable.
\end{lemma}
The proof is given in Appendix \ref{Lemma1proof}. This result allows us to define a change of variable which is given as:
\begin{definition}
On $[0,T] \times [\alpha,1]$, define,
\eqstar{
z(t,\xi) &:= -\log Q_\xi^{(0)} (t,\xi) + \frac{1}{2} \lambda_0^2 (T-t),\\
\psi(t) := -\log Q^{(0)}_\xi (t,\alpha) + \frac{1}{2} \lambda_0^2 & (T-t),\quad \varphi(t) := -\log Q^{(0)}_\xi (t,1) + \frac{1}{2} \lambda_0^2 (T-t),
}
and let
\[ q^{(0)}(t,z(t,\xi)) := Q^{(0)}(t,\xi). \]
\end{definition}
 It is clear from the boundary condition \eqref{eq:hjbzeroth_bound3} that we have $\varphi(t) = \infty \text{ for all }  0 \leq t < T.$ Then, we obtain the following PDE problem for $q^{(0)}(t,z)$. 
\begin{proposition}
\label{thm:zeroth const pde}
$q^{(0)}(t,z)$ satisfies the following linear PDE 
$$\Bigl(\frac{\partial}{\partial t} + \frac{1}{2} \lambda^2_0\frac{\partial}{\partial z^2}\Bigr) q^{(0)} = 0,  \text{ on } [0,T) \times (\psi(t), \infty),
$$
and the terminal and boundary conditions are 
\eqstar{
q^{(0)}(T,z) = U\Bigl( \bigl(U^\prime\bigr)^{-1}\bigl( \ee^{-z}\bigr)\Bigr), \quad
\lim_{z \to \infty }q^{(0)}_z(t,z) = 0, \quad
q^{(0)}(t,\psi(t)) = U\Bigl( \bigl(U^\prime\bigr)^{-1}\bigl( \ee^{-\psi(t) + \frac{\lambda_0^2}{2}(T-t)}\bigr)\Bigr).
}
\end{proposition}
The proof is given in Appendix \ref{Prop1proof}.

\begin{lemma}
\label{lemma:intermediate1}
Denote $q\bigl(t,z(t,\xi),x,y\bigr) := \hat{Q}(t,\xi,x, y).$ Then, on $[0,T) \times  (\psi(t), \infty) \times \R^2,$ we have
\eqstar{
\left( \frac{\del}{\del t} + \cA_0 + \cB_0 \right) \hat{Q} = \left( \frac{\del}{\del t} + \cA_0 + \cC_0 \right) q,
}
where 
\eqlnostar{eq:defC0}{
\cC_0 = \frac{1}{2} \lambda^2_0 \frac{\del^2}{\del z^2} + \rho \beta_0 \lambda_0  \frac{\del^2}{\del y \del z} + \sigma_0\lambda_0  \frac{\del^2}{\del x \del z}.
}
\end{lemma}
The above result follows from the calculations performed in the proof of Proposition \ref{thm:zeroth const pde} (also see \cite[Lemma 3.3]{lorig2015portfolio}). 

Next, we set $\hat{Q} = Q^{(0)}$ and $q = q^{(0)}$ in Lemma \ref{lemma:intermediate1}. Further we know that $q^{(0)}$ does not depend on $(x,y)$ and $\cA_0$ and the last two terms in $\cC_0$ have derivatives w.r.t.\ $(x,y).$ Then, we get the constant coefficient heat equation by applying the operator $\cC_0$ as in Proposition \ref{thm:zeroth const pde}. On $[0,T) \times (\psi(t),\infty)$, we have
\eqstar{
\left( \frac{\del}{\del t} + \cA_0 + \cC_0 \right) q^{(0)} = 0.
}
Finally, we define $q^{(1)}$ from $Q^{(1)}$ as
\begin{equation}
q^{(1)}(t,z(t,\xi),x,y) := Q^{(1)}(t,\xi,x,y).\label{q1def}
\end{equation}
\begin{proposition}
\label{thm:first const pde}
The alternative representation $q^{(1)}(t,z,x,y)$ of the first order correction term satisfies
\eqlnostar{eq:transformed_firstorder1}{
\left( \frac{\del}{\del t} + \cA_0 + \cC_0 \right) q^{(1)} + \cS_1 = 0, \text{ on } [0,T) \times (\psi(t), \infty) \times \R^2 ,
}
where
\eqlnostar{eq:main source term}{
\cS_1(t,z,x,y) = \bigl(\frac{1}{2} \lambda^2\bigr)_1(x,y) q^{(0)}_z(t,z,x,y).
}
The boundary conditions are 
\eqlnostar{eq:transformed_boundary1}{
q^{(1)}(T,z,x,y) = 0,\quad
q^{(1)}(t,\psi(t),x,y) = 0, \quad
\lim_{z\to \infty}q^{(1)}_z(t,z,x,y) = 0.
}
\end{proposition}
The above result follows from Definition \ref{def:Q1}. The solution to \eqref{eq:transformed_firstorder1} with boundary conditions \eqref{eq:transformed_boundary1} is given in terms of derivatives of $q^{(0)}$ in the following proposition.

\begin{proposition}
\label{prop:solution_firstorder}
The solution of the PDE in \eqref{eq:transformed_firstorder1} with boundary conditions \eqref{eq:transformed_boundary1} is given by 
\eqlnostar{eq:first correction term}{
q^{(1)}(t,z,x,y) = (T-t)\lambda_0 A(t,x,y)q^{(0)}_z(t,z) + \frac{1}{2}(T-t)^2 \lambda_0 B q^{(0)}_{zz}(t,z),
}
where
\eqstar{
A(t,x,y) &= \lambda_{1,0} \left[(x - \bar{x}) + \frac{1}{2}(T-t)b_0 \right] + \lambda_{0,1} \left[(y - \bar{y}) + \frac{1}{2}(T-t)c_0\right],\\
B &= \lambda_{1,0}\sigma_0\lambda_0 + \lambda_{0,1} \rho \beta_0 \lambda_0.
}
In the original variables, $Q^{(1)}$, the solution of \eqref{eq:nonlinear first order} with terminal and boundary conditions \eqref{eq:hjbbound11}, is given by
\eqlnostar{eq:first order final}{
Q^{(1)}(t,\xi,x,y) = (T-t)\lambda_0 A(t,x,y)\cD_1 Q^{(0)} +  \frac{1}{2}(T-t)^2 \lambda_0 B \left(\cD_3 -2 \cD_1\right)Q^{(0)}.
}
\end{proposition}
The proof is given in Appendix \ref{Prop4proof}.

 \subsubsection{Summary of the first order value function approximation results}\label{summary}
The coefficient polynomial approximation to the value function $Q$, solution to the PDE problem \eqref{eq:hjbmain2}-\eqref{fullbcns}
is then defined by setting $a=1$: $Q\approx Q^{(0)}+Q^{(1)}$, where
\begin{itemize}
\item Zeroth order term: $Q^{(0)}(t,\xi)$ is estimated by numerically solving \eqref{eq:hjbzeroth1} with the boundary conditions \eqref{eq:hjbzeroth_bound3}.
\item First order term: $Q^{(1)}(t,\xi,x,y)$ is obtained from 
Proposition \ref{prop:solution_firstorder} and is given by \eqref{eq:first order final}.
\end{itemize}

\subsection{Optimal strategy approximation}\label{optport}
\label{sec:optimal strategy}
Once we have the estimates for $Q^{(0)}$ and $Q^{(1)}$ in the approximate expansion \eqref{eq:approx solution} of the value function $Q$, we can find the first order approximation of the optimal strategy $\pi^*$ from the formula in \eqref{eq:optportfolio}. In terms of $Q(t,\xi,x,y),$ the optimal strategy is given by 
\eqstar{
\pi^{*}(t,l,m,x,y) = -m \left[ \frac{(\mu(x)-r) Q_\xi}{(\sigma(x,y))^2 Q_{\xi \xi}} + \frac{\rho \beta(y) Q_{y \xi}}{\sigma(x,y) Q_{\xi \xi}} + \frac{Q_{x \xi }}{Q_{\xi \xi}} \right], \text{ with } \xi = \frac{l}{m}.
}
To express the approximation for $\pi^*$ in terms of $\cR, Q^{(0)}$ and their spatial derivatives, we first replace $Q$ by $Q^{(0)}+Q^{(1)}$ in the above formula, use the results in \eqref{eq:first order final} and the following Lemma. 
\begin{lemma}
\label{Rcalcs}
From the definition \eqref{eq:define risk tolerance} of $\cR$, we have the following identities:
\eqstar{
&\text{(i) } (\cD_1 + \cD_2 )\cD_1 Q^{(0)} = \cR\cR_{\xi\xi} \cD_1Q^{(0)},\\
&\text{(ii) } (-2\cD_1 + \cD_3) Q^{(0)} = \cD_1 \cD_1 Q^{(0)},\\
&\text{(iii) } \bigl(\cD_1   +\cD_2 \bigr)  \cD_1 \cD_1 Q^{(0)} = \cR\bigl(\cR_{\xi \xi} (3 \cR_\xi -2) + \cR \cR_{\xi \xi \xi} \bigr) \cD_1 Q^{(0)}.
}
\end{lemma}
\begin{proof}
We show the following using elementary manipulations. From \eqref{eq:define risk tolerance} and \eqref{eq:def diff operator}, recall that
\eqstar{
\cR = -\frac{Q^{(0)}_\xi}{Q^{(0)}_{\xi\xi}}, \quad
\cD_k = \cR^k \frac{\del^k}{\del \xi^k}, \quad k = 1,2,\ldots.
}

\noindent (i) We have,
\eqstar{
\cD_1 \cD_1 Q^{(0)} &= \cD_1 (\cR Q^{(0)}_{\xi}) = \cR \cR_{\xi} Q^{(0)}_{\xi} + \cR^2 Q^{(0)}_{\xi \xi} =  (\cR_\xi -1)\cD_1Q^{(0)}, \text{ and}\\
\cD_2 \cD_1 Q^{(0)} &= \cR\cR_{\xi\xi} \cD_1Q^{(0)} - (\cR_\xi -1) \cD_1Q^{(0)}.
}
The above result and the distributive property of $\cD_k$ operator completes the proof. 

\noindent (ii) We have,
\eqstar{
\cD_3 Q^{(0)} &= \cR^3 \del_\xi \Bigl(-\frac{Q^{(0)}_\xi}{\cR}\Bigr) = \cR^3\Bigl(-\frac{Q^{(0)}_{\xi\xi}}{\cR} + \frac{Q^{(0)}_\xi \cR_\xi}{\cR^2}\Bigr) = (\cR_{\xi} + 1)\cD_1 Q^{(0)}.
}
This gives,
\eqstar{
-2\cD_1Q^{(0)} + \cD_3 Q^{(0)} &= -2\cD_1Q^{(0)} + (\cR_{\xi} + 1)\cD_1 Q^{(0)} = (\cR_\xi - 1) \cD_1 Q^{(0)}.
}
The final conclusion follows from (i).

\noindent (iii) Using the previous calculations, we get
\eqstar{
\cD_1\bigl( (\cR_\xi - 1) \cD_1 Q^{(0)} \bigr) &= \cR^2 \cR_{\xi \xi}Q^{(0)}_\xi + (\cR_\xi - 1) \cD_1 \cD_1 Q^{(0)} =  \cR \cR_{\xi \xi} \cD_1 Q^{(0)}  + (\cR_\xi - 1)^2 \cD_1 Q^{(0)},\\
\cD_2 \bigl( (\cR_\xi - 1) \cD_1 Q^{(0)} \bigr) &= \cR \cD_1 \bigl( \cR \cR_{\xi \xi} Q^{(0)}_\xi + (\cR_\xi - 1)^2 Q^{(0)}_\xi \bigr)\\
&= \cR\bigl( \cR \cR_{\xi \xi \xi} + \cR_{\xi \xi} (\cR_\xi - 1)\bigr) \cD_1Q^{(0)} + \cR \cD_1 \bigl((\cR_\xi - 1)^2 Q^{(0)}_\xi \bigr)\\
&= \cR\bigl( \cR\cR_{\xi \xi \xi}  + 3 \cR_{\xi \xi} (\cR_\xi - 1)\bigr) \cD_1Q^{(0)} -  (\cR_\xi - 1)^2 \cD_1Q^{(0)}.
}
The sum of above two results concludes the proof. 
\end{proof}

Thus, we obtain the optimal strategy approximation as 
\eqlnostar{eq:optimal strategy final}{
\pi^* &\approx m \left[\frac{(\mu(x,y)-r)}{(\sigma(x,y))^2} \cR + (T-t)\lambda_0 A(t,x,y) \frac{(\mu(x,y)-r)}{(\sigma(x,y))^2}  \cR^2\cR_{\xi\xi} \right. \nonumber\\
&+ \left.   \frac{1}{2}(T-t)^2 \lambda_0 B \frac{(\mu(x,y)-r)}{(\sigma(x,y))^2}  \cR^2\bigl(  \cR_{\xi \xi} (3 \cR_\xi -2) + \cR \cR_{\xi \xi \xi} \bigr) \right.\\
&+\left.  (T-t) \lambda_0 \Bigl(\frac{\lambda_{0,1}\rho \beta(y)}{\sigma(x,y)}+ \lambda_{1,0} \Bigr) \cR (\cR_\xi - 1) \right].
}

\section{Examples and Numerical Implementation}
\label{sec:numerics}
In this section, we consider the stochastic volatility model as in  \citet{chacko2005dynamic} with their calibrated set of parameters and provide a detailed discussion of the numerical implementation of our results obtained in Section \ref{sec:main results}. We discuss the effect of stochastic volatility on value function and optimal strategy for the case of \emph{power} utility function and a \emph{mixture} of two power utility functions, as introduced in \cite{fouque2015portfolio}. The latter allows for relative aversion that declines with wealth, while for the former it is constant across wealth levels.

Under the considered stochastic volatility model \cite[Section 1]{chacko2005dynamic}, the coefficients $(\mu, \sigma, c, \beta)$ in Section \ref{sec:problem formulation} are independent of $x$ and are given as 
\eqstar{
\mu(y) = \mu, \quad \sigma(y) = \frac{1}{\sqrt{y}}, \quad c(y) = \kappa (\theta - y), \quad \beta(y) = \delta \sqrt{y}.
}
The market calibrated values of the constants are 

\begin{table}[H]
\centering
\small
\begin{tabular}{c c c c c c}
\toprule
$\mu-r$ & $\kappa$ & $\theta$ & $\delta$ & $\rho$\\
\toprule
0.0811 & 0.3374 & 27.9345 & 0.6503 & 0.5241\\
\hline
\end{tabular}
\label{tab:parval}
\end{table}
We numerically solve for $Q^{(0)}$ backward in time via explicit finite-difference Euler scheme. We approximate the domain $[0,T] \times [\alpha,1]$ with a uniform mesh given as 
\eqstar{
\cM = \bigl\{(t^n,\xi_j): n = 0,1,\ldots, N, \quad j = 0,1,\ldots,J\bigr\},
}
where $t^n = T - n \Delta t, \xi_j = \alpha + j \Delta \xi.$ Let $Q^n_j$ denote the numerical approximation of $Q^{(0)}(t^n,\xi_j)$.
Then the discretized equation for $Q^{(0)}$ in the interior is written as
\eqlnostar{eq:fin diff scheme}{
Q^{n+1}_j = Q^n_j - \frac{1}{8 }\lambda^2_0 \Delta t\frac{(Q^n_{j+1} - Q^n_{j-1})^2}{(Q^n_{j+1} - 2 Q^n_j + Q^n_{j-1})}.
}
We start with the guess $Q^0_j = U(\xi_j), \text{ for all } j = 0, 1, \ldots, J$, and the boundary conditions are
\eqlnostar{eq:fin diff bound}{
Q^{n+1}_J = Q^{n+1}_{J-1}, \text{ and, } 
Q^{n+1}_0 = U(\xi_0).
}
\begin{figure}[htbp]
\centering
\subfigure[]{\label{fig:chacko_Q0}\includegraphics[trim = 130 280 130 280, width=0.48\textwidth]{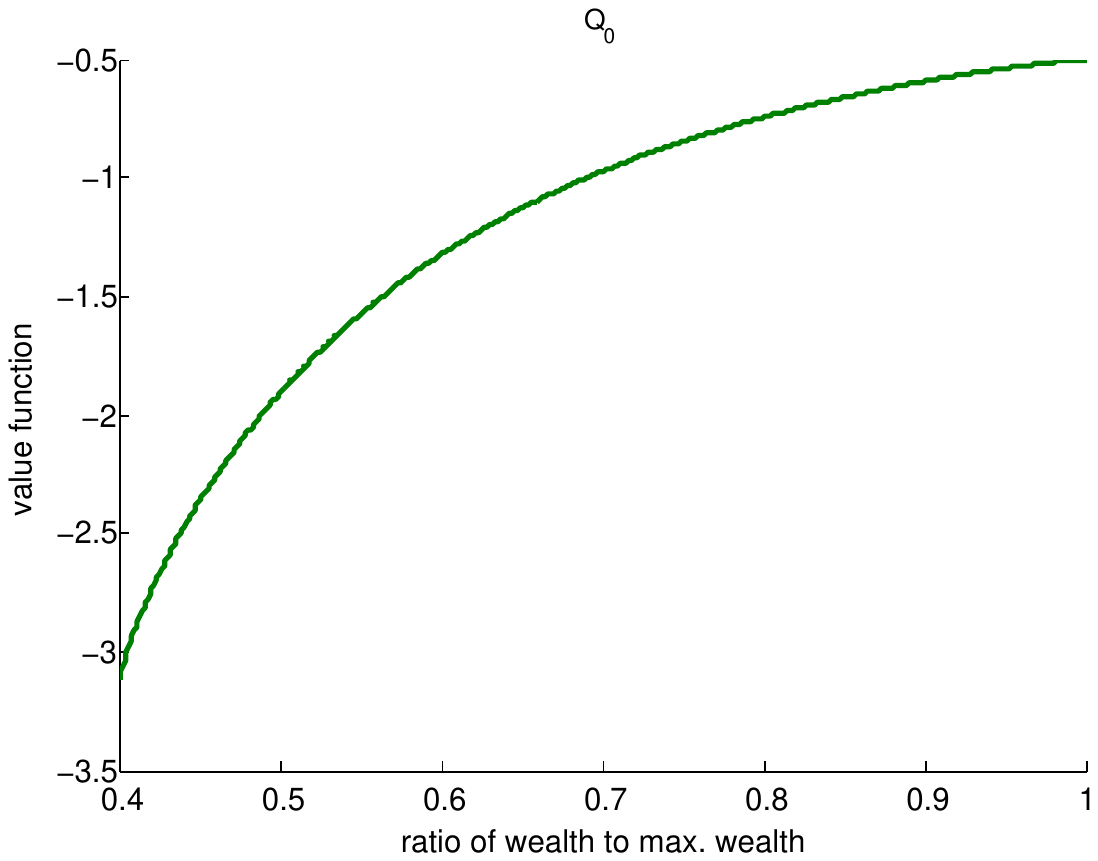}}
~
\subfigure[]{\label{fig:chacko_Qloss}\includegraphics[trim = 130 280 130 280, width=0.48\textwidth]{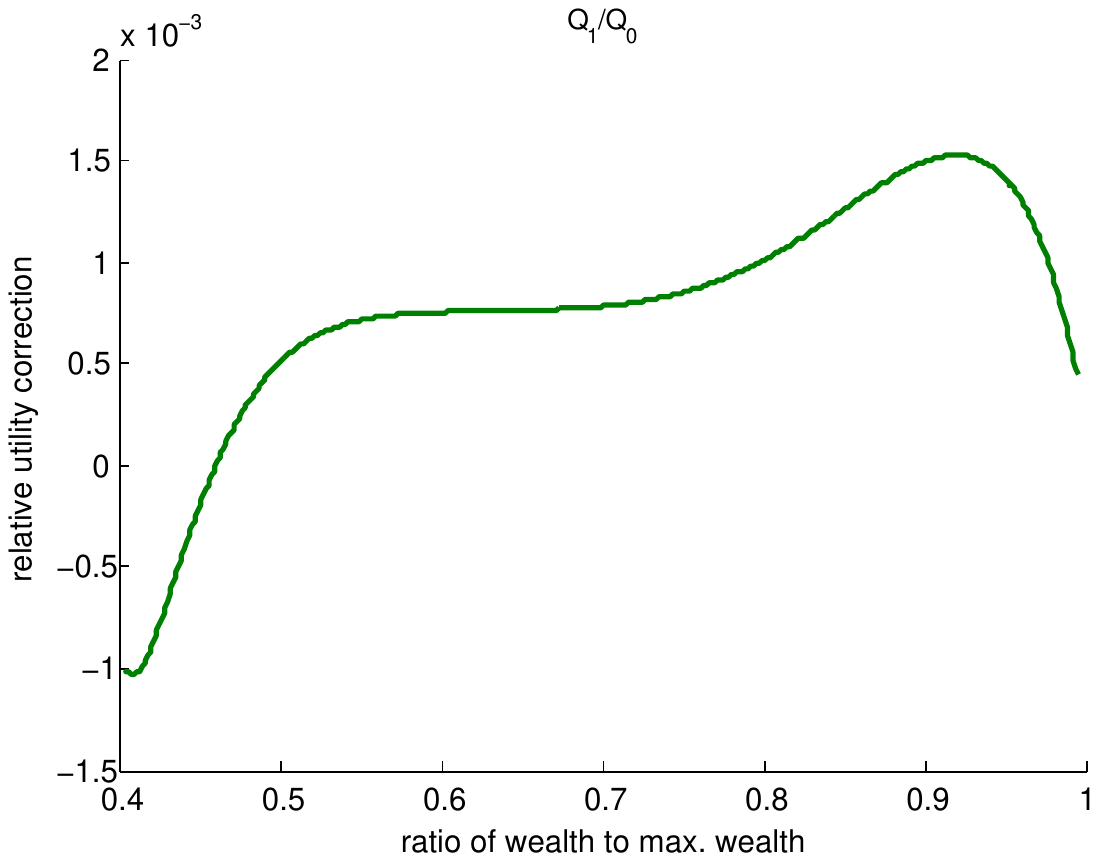}}
\caption{{\small{{\em Numerical solutions to (a) zeroth order value function $Q^{(0)}$ (b) relative utility correction $Q^{(1)}/Q^{(0)}.$ Utility function used $U(\xi) = \frac{\xi^{1-\gamma}}{1-\gamma}, \gamma = 3.0$}}}}
\end{figure}
In Figure \ref{fig:chacko_Q0} and \ref{fig:chacko_mix_Q0}, we plot the numerical solution for the leading order expansion term $Q^{(0)}$ obtained from \eqref{eq:fin diff scheme}.  We can see that the zeroth order term is concave and non-decreasing as expected from Lemma \ref{lem:zeroth order properties}.  
\begin{figure}[htbp]
\centering
\subfigure[]{\label{fig:chacko_mix_Q0}\includegraphics[trim = 130 280 130 280, width=0.48\textwidth]{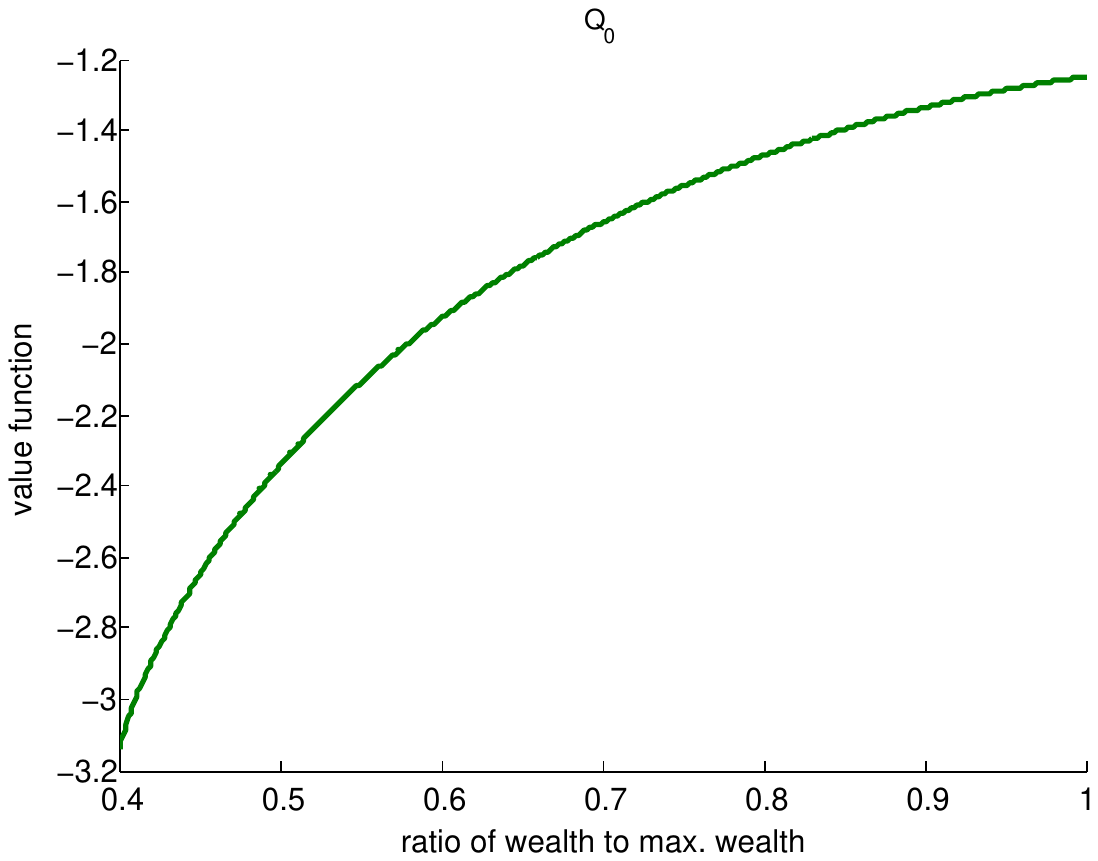}}
~
\subfigure[]{\label{fig:chacko_mix_Qloss}\includegraphics[trim = 130 280 130 280, width=0.48\textwidth]{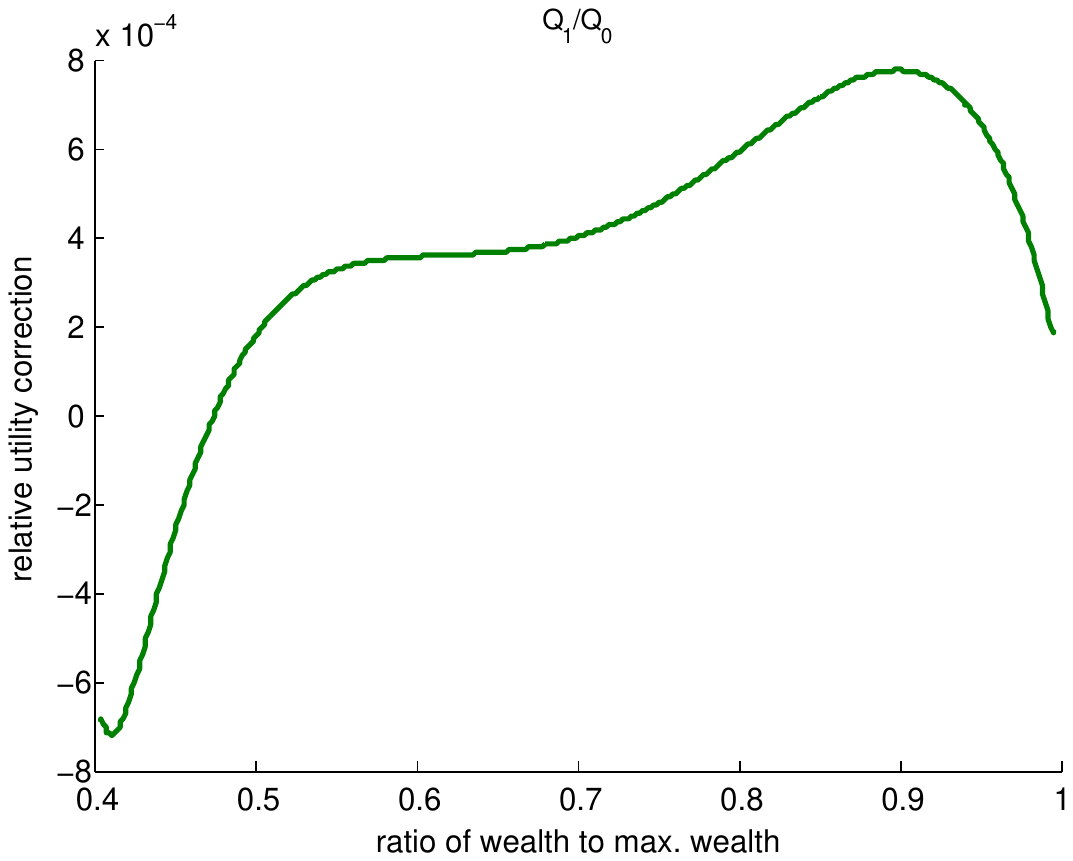}}
\caption{{\small{{\em Numerical solutions to (a) zeroth order value function $Q^{(0)}$ (b) relative utility correction $Q^{(1)}/Q^{(0)}.$ Utility function used $U(\xi) = \frac{\xi^{1-\gamma_1}}{1-\gamma_1} + \frac{\xi^{1-\gamma_2}}{1-\gamma_2}, \gamma_1 = 3.0, \gamma_2 = 1.5.$}}}}
\end{figure}

To find the first order correction term, we refer to formula \eqref{eq:first order final}. We can directly use $\cR$ from Proposition \ref{thm:risk tolerance} in the formula instead of taking derivatives of $Q^{(0)}.$ We note that to obtain $Q^{(1)}$, we need to set the value for reference level $\bar{y}.$ We set $\bar{y} = y,$ the current value of the stochastic volatility factor. This gives us
\eqlnostar{eq:first order chacko}{
Q^{(1)} &= \Bigl(\frac{1}{2} \lambda^2 \Bigr)_{0,1} (T-t)^2 c_0 \cR Q^{(0)}_{\xi} +  \Bigl(\frac{1}{2} \lambda^2 \Bigr)_{0,1} (T-t)^2  \lambda_0 \rho \beta_0 \Bigl( -2  \cR Q^{(0)}_{\xi} + \cR^3 \partial^3_\xi Q^{(0)} \Bigr)\\
&= \frac{1}{2}(\mu-r)^2 (T-t)^2\Bigl[ \kappa (\theta - y) \cR Q^{(0)}_{\xi} + \rho \delta (\mu-r)y\Bigl( -2  \cR Q^{(0)}_{\xi} + \cR^3 \partial^3_\xi Q^{(0)} \Bigr) \Bigr].
}
We use the regularity properties of $\cR$ and $Q^{(0)}$ to compute the above expression. We obtain estimates of $\cR$ by numerically solving \eqref{eq:risk tol pde} with boundary conditions \eqref{eq:R_bound1} via explicit finite-difference Euler scheme. The discretized equation in the interior is written as 
\eqlnostar{eq:fin diff R}{
R^{n+1}_j = R^n_j + \frac{1}{2}\lambda^2_0 \Delta t (R^n_j)^2 \frac{(R^n_{j+1} - 2 R^n_j + R^n_{j-1})}{(\Delta \xi)^2},
}
and the boundary conditions
$R^{n+1}_J = 0$, and  
$R^{n+1}_0 = 0$.
As we solve the scheme backward in time, we start with the guess $R^0_j = -\frac{U^\prime(\xi_j)}{U^{\prime\prime}(\xi_j)}, \text{ for all } j = 0, 1,\ldots, J.$ In our market calibrated stochastic volatility model, we set $y = \theta$ and plot the relative utility correction in Figure \ref{fig:chacko_Qloss} and Figure \ref{fig:chacko_mix_Qloss}. We observe that the change in the value function due to the introduction of stochastic volatility is negligible.

Next, we calculate the approximation to optimal strategy whose different terms are given from \eqref{eq:optimal strategy final} as 
\eqstar{
\frac{\pi^*_0}{m} &=  (\mu-r) y \cR,\\
\frac{\pi^*_1}{m} &=   \frac{(\mu-r)^3y^2}{2}(T-t)^2 \Bigl[\kappa (\theta - y) \Bigl( \cR^2 Q^{(0)}_{\xi \xi}\Bigr)\\
&+ \rho \delta \Bigl( \cR^2 \cR_{\xi \xi} (3 \cR_\xi -2) + \cR^3 \cR_{\xi \xi \xi} \Bigr)\Bigr]+ (\mu-r)^2(T-t) \rho \delta y\cR \Bigl(\cR_\xi - 1 \Bigr).
}
\begin{figure}[htbp]
\centering
\subfigure[]{\label{fig:chacko_piapprox}\includegraphics[trim = 130 280 130 280, width=0.48\textwidth]{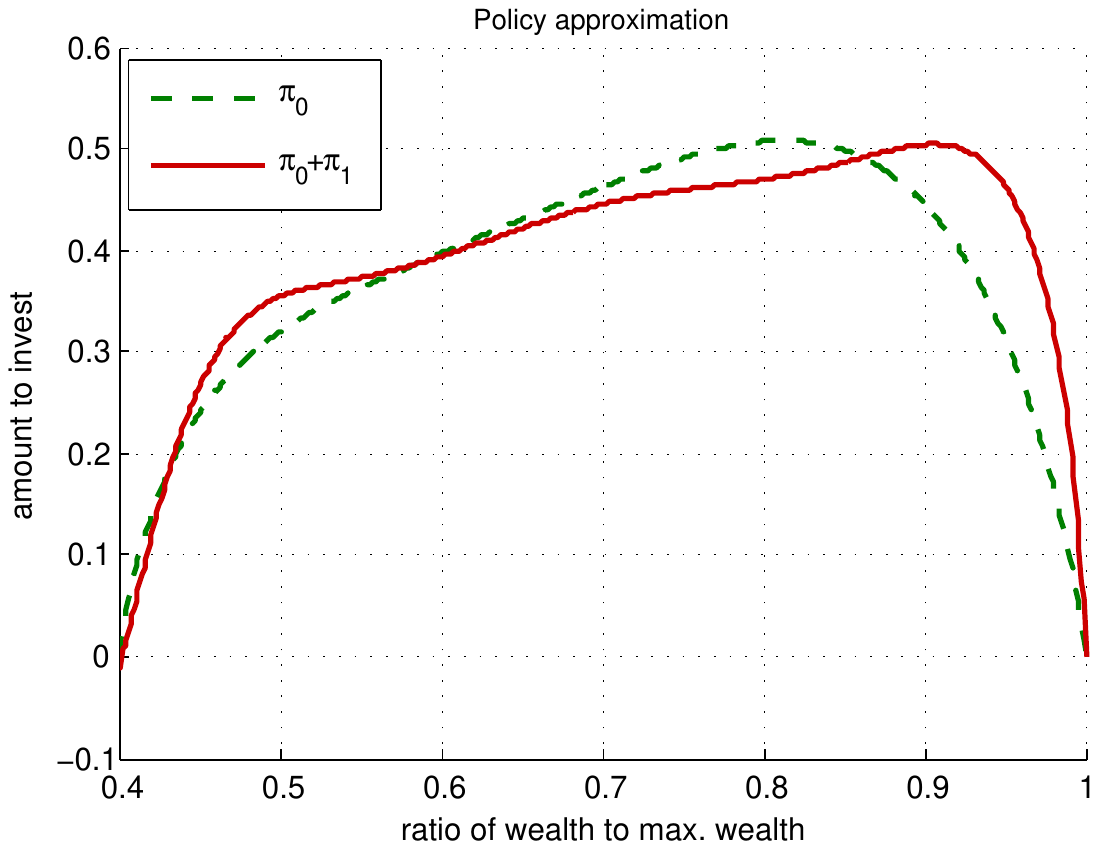}}
~
\subfigure[]{\label{fig:chacko_mix_piapprox}\includegraphics[trim = 130 280 130 280, width=0.48\textwidth]{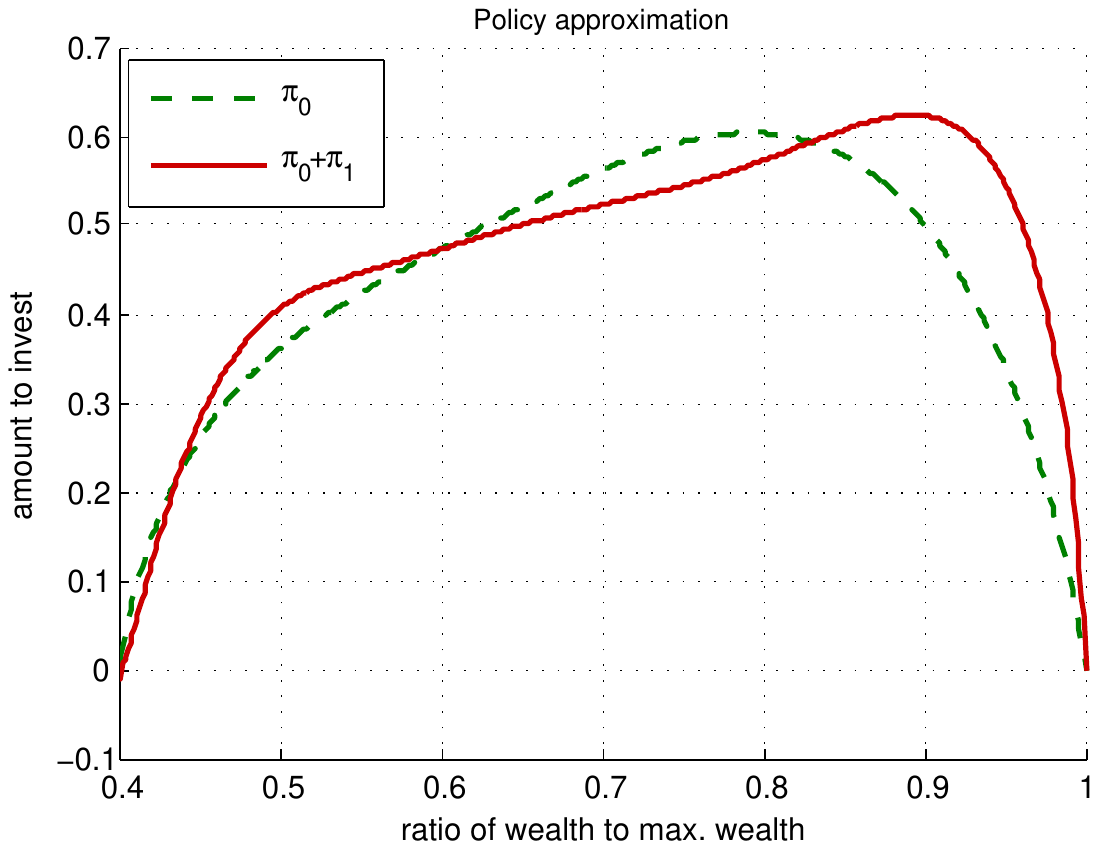}}
\caption{{\small{{\em Numerical solutions to the optimal strategy approximation for utility function (a) $U(\xi) = \frac{\xi^{1-\gamma}}{1-\gamma}, \gamma = 3.0$ (b) $U(\xi) = \frac{\xi^{1-\gamma_1}}{1-\gamma_1} + \frac{\xi^{1-\gamma_2}}{1-\gamma_2}, \gamma_1 = 3.0, \gamma_2 = 1.5.$}}}}
\end{figure}
We suppose that the initial value of maximum wealth is unity, i.e.\ we set $m=1.0$ and plot numerical solution to the leading order term $\pi_0$ and to the first order approximation $\pi_0 + \pi_1$ in Figure \ref{fig:chacko_piapprox} and \ref{fig:chacko_mix_piapprox}. It is interesting to note that to achieve similar value functions without and with the stochastic volatility correction, i.e.\ $Q^{(0)}$ and $Q^{(0)} + Q^{(1)},$ we clearly need to employ two very different investment policies, namely $\pi_0$ and $\pi_0 + \pi_1.$

In Figure \ref{fig:chacko_piapprox} and \ref{fig:chacko_mix_piapprox}, we note that as the current wealth approaches to the maximum wealth value, the optimal strategy is to gradually liquidate the position in the risky asset. In the presence of stochastic volatility, the optimal strategy approximation $\pi_0 + \pi_1$ suggests to hold the risky position longer than without the stochastic volatility correction as in $\pi_0.$ The corrected strategy also suggests to sharply liquidate the position in the risky asset to safeguard from the downside risk of stochastic volatility. On the other hand, when the current wealth moves away from the drawdown barrier, the optimal strategy approximation $\pi_0 + \pi_1$ suggests to build up a position in the risky asset at about the same trading rate to that in the case of constant volatility approximation $\pi_0.$ 

From the above results, we deduce that even in the presence of stochastic volatility, the investor does not lose much value in his portfolio. However, to achieve similar value functions, the investor has to deploy a remarkably different strategy corrected for stochastic volatility $\pi_0 + \pi_1$ when compared to the constant volatility strategy $\pi_0.$ The larger position in the risky asset when moving away from the drawdown barrier suggests leveraging the possible upside due to stochastic volatility while holding on to the risky asset longer than in the constant volatility case when close to the optimal level suggests caution towards a possible downside risk. 

In the above results, we have set the level of stochastic volatility factor $y$ to be the same as the long term value $\theta.$ As it is clear that the level of stochastic volatility plays a crucial role in the correction terms, we studied the effects when $y$ moves in either direction away from its long term value $\theta.$ We observed that even in these new cases, the relative utility correction remains small. However, the optimal strategy in these cases exhibit remarkably different behaviors. When the current level of volatility is higher than the long-term average $y = 1.05\times \theta$, in Figure \ref{fig:chacko_high_piapprox} the optimal strategy approximation suggests to invest more in the risky asset compared to the strategy without stochastic volatility correction. Also, as the portfolio wealth moves away the drawdown barrier, the corrected optimal strategy suggests to build up the position in risky asset at a much higher rate than suggested by $\pi_0.$ Whereas, in the case when the current level of volatility is lower than the long-term average $y = 0.95\times \theta$, in Figure \ref{fig:chacko_low_piapprox} the optimal strategy approximation suggests to invest less in the risky asset compared to the strategy without stochastic volatility correction. Still close to the maximum wealth value, the corrected strategy suggests to hold more risky asset than the constant volatility strategy suggests.
\begin{figure}[htbp]
\centering
\subfigure[]{\label{fig:chacko_high_piapprox}\includegraphics[trim = 130 280 130 280, width=0.48\textwidth]{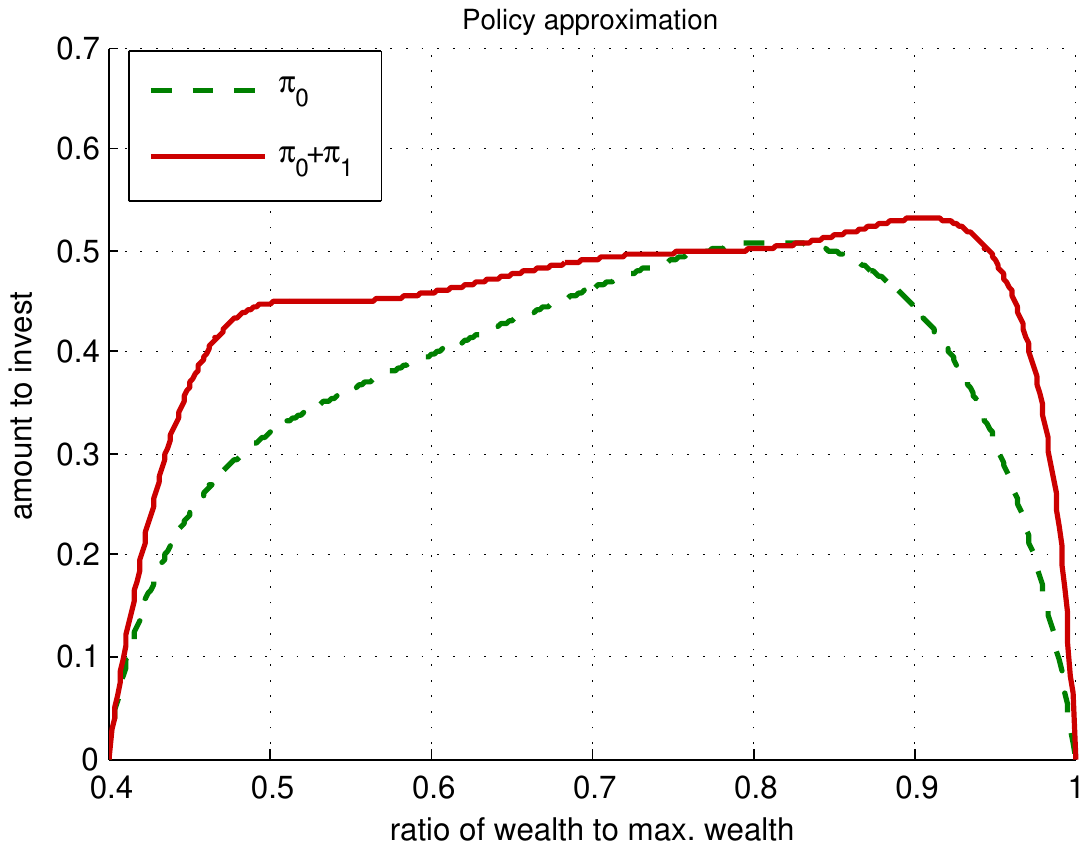}}
~
\subfigure[]{\label{fig:chacko_low_piapprox}\includegraphics[trim = 130 280 130 280, width=0.48\textwidth]{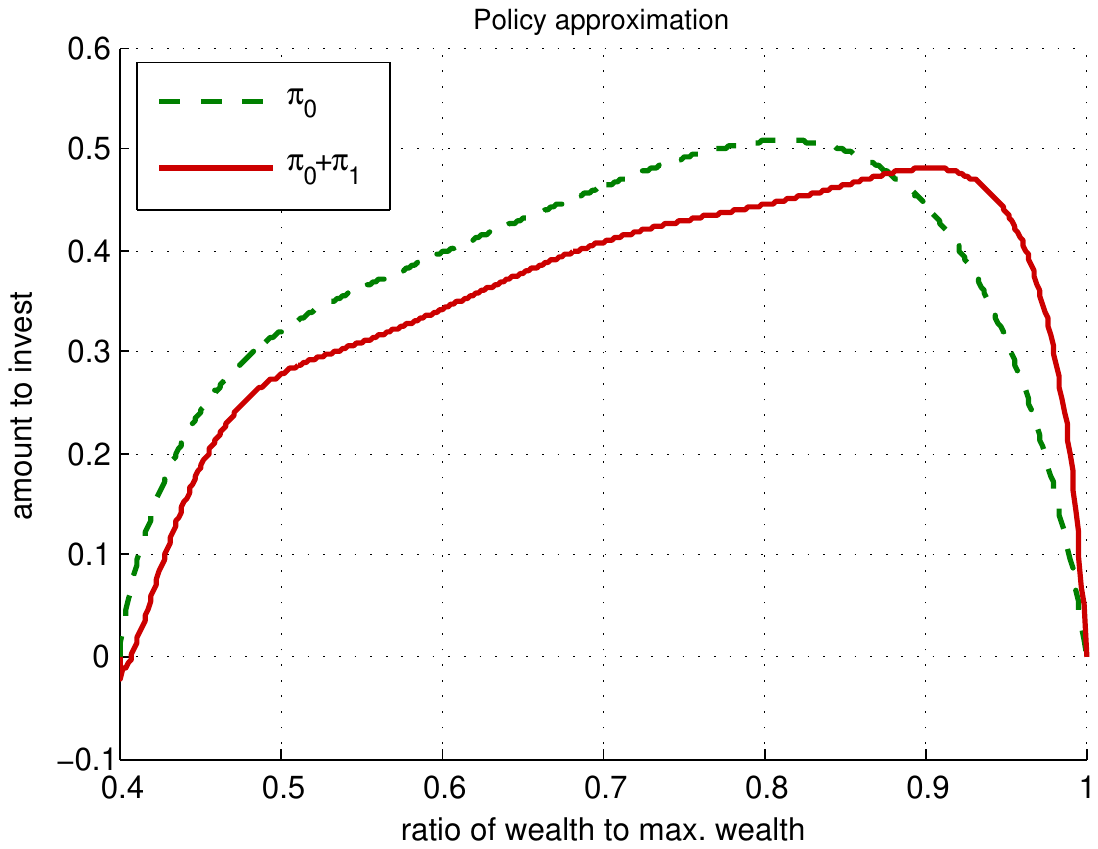}}
\caption{{\small{{\em Numerical solutions to the optimal strategy approximation for (a) $y = 1.05 \times \theta$ (b) $y = 0.95 \times \theta.$ Utility function used $U(\xi) = \frac{\xi^{1-\gamma}}{1-\gamma}, \gamma = 3.0$}}}}
\end{figure}

\section{Conclusion}
\label{sec:conclusion}
We studied the impact of stochastic Sharpe ratio in a dynamic portfolio optimization problem under a drawdown constraint. We proposed a new investor objective framework which allows for a dimensionality reducing transformation. This new setting allowed us to employ coefficient expansion technique to solve for different terms in the approximation of the value function and optimal strategy. With the help of a nonlinear transformation we derived value function expansion terms which can be numerically calculated and used to approximate the optimal portfolio strategy. In a popular stochastic volatility model with market calibrated parameters, we illustrated the remarkable differences between optimal strategies with and without stochastic volatility correction. 

The current problem requires further investigation which can be performed along the following directions: 
\begin{enumerate}
\item Approximation error analysis: In this work, we focussed our attention to capture the first order effects of stochastic volatility on value function and optimal portfolio strategy. We observed that the stochastic volatility correction to value function is small whereas the corrected optimal strategy exhibited remarkably different behavior than the constant volatility optimal strategy. This calls for an investigation of the higher order terms to look for possible other interesting effects on the optimal strategy. 
\item Multi-asset market model: We studied the portfolio optimization problem under drawdown constraint in a stochastic volatility model which provides a sensible guide towards informed investment decisions. However, in order to completely capture the market conditions, we plan to tackle the same problem in a multi-asset model setting and study the effect of stochastic volatility on investment strategies. 
\end{enumerate}

\appendix

\section{Proofs}
\label{sec:proofs}
\subsection{Proof of Lemma \ref{lem:zeroth order properties}}\label{Lemma1proof}
\begin{proof}
Let us consider a market with a risky asset whose dynamics is given by a lognormal model
$$ \frac{\dd S_t}{S_t} = \mu_0 \dd t  + \sigma_0 \dd B^{(1)}_t.$$
With this risky asset in the market, we once again formulate our investor's portfolio optimization problem (see Section \ref{sec:discounted formulation})
\eqstar{
\cV(t,l,m) &= \sup_{\pi \in \Pi}\E\Bigl[ U\Bigl(\frac{L_T}{M_T}\Bigr) \Big \vert L_t = l, M_t = m\Bigr], \quad t > 0, m > l > \alpha m > 0,
}
where the admissible strategies are given by
\eqstar{
\Pi_{\alpha,t,l,m} := \bigl\{ \pi: \text{measurable}, \, \F-\text{adapted}, \, &\E_{t,l,m}\int^T_t \pi^2_s \dd s < \infty \text{ s.t. } L_s \geq \alpha M_s > 0\text{ a.s.}, t \leq s \leq T \bigr\},
} 
and $\mathbb{F} = \{\cF_t: 0 \leq t \leq T \}$ is  the augmentation of the filtration generated by $B^{(1)}.$ We define the constant Sharpe ratio as $\lambda_0 := \frac{(\mu_0 - r )}{\sigma_0}$ and the space domain as $\crO_{\alpha} :=\{(l,m): m > l >\alpha m > 0\} \subset \R^2.$ Then, by proceeding as in Section \ref{sec:discounted formulation}, it can be shown that for $\cV \in C^{1,2,1}(\R_+ \times \crO_{\alpha}),$ we have the following nonlinear PDE
\eqstar{
\del_t \cV - \frac{1}{2}\lambda^2_0\frac{(\cV_l)^2}{\cV_{ll}} = 0, \text{ on } [0,T) \times \crO_{\alpha},
}
and the boundary conditions are
\eqstar{
\cV(T,l,m) = U(l/m),\quad
\cV_m(t,m,m) = 0,\quad
\cV(t,\alpha m,m) = U(\alpha).
}
Similar to Section \ref{sec:dimension reduction}, we perform a change of variable $\xi := l/m.$ It is then clear that the leading order term in expansion \eqref{eq:approx solution}, $Q^{(0)}(t,\xi) = \cV(t,l,m).$ 

To first show that $Q^{(0)}(t,\cdot)$ is a non-decreasing function, we recall that in the constant volatility model, for a portfolio strategy $\pi,$ the discounted wealth process is given as 
\eqstar{
L^{l,\pi}_t = l + \int^t_0 \pi_s \sigma_0 \bigl( \lambda_0 \dd s +  \dd B^{(1)}_s \bigr),
}
where $l$ is the starting wealth value. Let $(L^{l,\pi})^*$ denote the maximum of wealth process $L^{l,\pi}$ over the time period $[0,T].$ Now, we consider $l,l^\prime$ for a fixed value of $m$ such that $(t,l,m), (t,l^\prime,m) \in [0,T) \times \crO_\alpha.$ Then, for $l \leq l^\prime,$ we choose $\pi \in \Pi_{\alpha,t,l,m}$ such that we have 
\eqlnostar{eq:increasing1}{
L^{l,\pi} &\geq \alpha (m \vee (L^{l,\pi})^*)\\
&= \alpha \Biggl(m \vee \Bigl(l + \bigl(\int^t_0 \pi_s \sigma_0 ( \lambda_0 \dd s +  \dd B^{(1)}_s )\bigr)^* \Bigr) \Biggr).
}
Add $(l^\prime - l)$ to both sides of the inequality above to write
\eqstar{
L^{l^\prime,\pi} &\geq  \Biggl(\bigl(\alpha m + (l^\prime - l) \bigr) \vee \Bigl(\alpha l^\prime + \alpha \bigl(\int^t_0 \pi_s \sigma_0 ( \lambda_0 \dd s +  \dd B^{(1)}_s )\bigr)^*  + (1-\alpha)(l^\prime -l )\Bigr) \Biggr)\\
&\geq \alpha\Biggl( m \vee \Bigl( l^\prime +  \bigl(\int^t_0 \pi_s \sigma_0 ( \lambda_0 \dd s +  \dd B^{(1)}_s )\bigr)^*\Bigr) \Biggr)\\
&= \alpha (m \vee (L^{l^\prime,\pi})^*),
}
which gives that $\Pi_{\alpha,t,l,m} \subset \Pi_{\alpha,t,l^\prime,m}.$ Thus, we get $\cV(t,l,m) \leq \cV(t,l^\prime,m).$ For $ \xi:= \frac{l}{m}$ and $\xi^\prime :=\frac{l^\prime}{m},$ this gives us 
$$Q^{(0)}(t,\xi) \leq Q^{(0)}(t,\xi^\prime).$$
Next, it follows from the arguments presented in Lemma 3.2 \citet{elie2008finite} that $\cV(t,l,m)$ is non-increasing in variable $m.$ Thus, for fixed $l$ and $m \leq m^\prime$ such that $(t,l,m), (t,l,m^\prime) \in [0,T) \times \crO_\alpha,$  we have $\cV(t,l,m^\prime) \leq \cV(t,l,m).$ Once again by defining $\xi^\prime := \frac{l}{m^\prime}$ and $\xi := \frac{l}{m},$ we get 
\eqstar{
\cV(t,l,m^\prime) \leq \cV(t,l,m) \implies Q^{(0)}(t,\xi^\prime) \leq Q^{(0)}(t,\xi).
}
Therefore, we have shown that $Q^{(0)}(t,\cdot)$ is non-decreasing.

In order to show concavity of value function $Q^{(0)}(t,\cdot),$ we take motivation from the arguments presented in Lemma 3.2 \citet{elie2008finite}. First, we fix $\eta \in [0,1]$ and choose $\alpha \leq \xi_1,\xi_2 \leq 1.$ Our aim is to show that $\cV(t,l,m)$ is concave in its second argument, i.e.\
\eqlnostar{eq:claim1}{
\eta \cV(t,l_1,m) + (1-\eta) \cV(t,l_2,m)  \leq \cV(t, \eta l_1 + (1-\eta) l_2,m),
}
where for a fixed value of $m,$ we set $l_1 = m \xi_1$ and $l_2 = m \xi_2.$ Now, suppose \eqref{eq:claim1} is true. Then by reversing the change of variables, we get in \eqref{eq:claim1}
\eqstar{
\eta Q^{(0)}(t,\xi_1) + (1-\eta) Q^{(0)}(t,\xi_2) \leq Q^{(0)}(t,\eta \xi_1 + (1-\eta) \xi_2)
}
which gives us concavity of $Q^{(0)}(t,\cdot).$ It remains to show that \eqref{eq:claim1} is indeed true. 
 
We define process $L^{(1)}$ as the wealth process with starting wealth $l_1$ and portfolio strategy $\pi_1 \in \Pi_{\alpha,t,l_1,m}.$ Similarly, we define the process $L^{(2)}$ with starting wealth $l_2$ and portfolio strategy $\pi_2 \in \Pi_{\alpha,t,l_2,m}.$  Then, we have by definition
\eqstar{
\eta L^{(1)} + (1-\eta)L^{(2)} &\geq \eta \alpha (m \vee (L^{(1)})^*) + (1-\eta) \alpha (m \vee (L^{(2)})^*)\\
& \geq \alpha\Bigl(m \vee \bigl(\eta L^{(1)} + (1-\beta) L^{(2)}\bigr)^*\Bigr).
}
This gives us that $\eta \pi_1 + (1-\eta) \pi_2 \in \Pi_{\alpha,t,\eta l_1 + (1-\eta)l_2,m}.$ From concavity property of utility function $U$, it follows
\eqstar{
&\eta \E_t\left[U\Biggl(\frac{L^{(1)}_T}{(m \vee (L^{(1)})^*_T)}\Biggr) \right]+ (1-\eta) \E_t\left[U\Biggl(\frac{L^{(1)}_T}{(m \vee (L^{(1)})^*_T)}\Biggr) \right] \\
&\leq \E_t\left[U\Biggl(\frac{\eta L^{(1)}_T}{(m \vee (L^{(1)})^*_T)} + \frac{(1-\eta) L^{(2)}_T}{(m \vee (L^{(2)})^*_T)}\Biggr) \right].
}
Next, we intend to show that
\eqlnostar{eq:final step concavity}{
\frac{\eta L^{(1)}_T}{(m \vee (L^{(1)})^*_T)} + \frac{(1-\eta) L^{(2)}_T}{(m \vee (L^{(2)})^*_T)} \leq \frac{\eta L^{(1)}_T + (1-\eta) L^{(2)}_T}{(m \vee (\eta L^{(1)} + (1-\eta) L^{(2)})^*_T)}.
}
Consider the following possible scenarios where we compare the respective terms with $m$ and find the maximum
\begin{table}[H]
\centering
\small
\begin{tabular}{c c c c}
\toprule
       & $(L^{(1)})^*_T$ & $(L^{(2)})^*_T$ & $(\eta L^{(1)} + (1-\eta) L^{(2)})^*_T$\\
\hline
Case 1 & $m$ & $m$ & $m$\\
Case 2 & $m$ & $(L^{(2)})^*_T$ & $m$\\
Case 3 & $(L^{(1)})^*_T$ & $m$ & $m$\\
Case 4 & $(L^{(1)})^*_T$ & $(L^{(2)})^*_T$ & --\\
\hline
\end{tabular}
\end{table}
It is clear that the inequality in \eqref{eq:final step concavity} holds for Case 1--3 and we only need to consider Case 4. We know from the optimality condition that for strategies $\pi_1$ and $\pi_2$ which attain the maximum, the position in the risky asset becomes zero thereafter as the maximum possible utility is achieved. It follows that for such strategies, we have
\eqstar{
L^{(1)}_T = (L^{(1)})^*_T, \qquad L^{(2)}_T = (L^{(2)})^*_T.
}
Then, we get 
\eqstar{
\frac{\eta L^{(1)}_T + (1-\eta) L^{(2)}_T}{(m \vee (\eta L^{(1)} + (1-\eta) L^{(2)})^*_T)} = \frac{\eta (L^{(1)})^*_T + (1-\eta) (L^{(2)})^*_T}{(m \vee (\eta L^{(1)} + (1-\eta) L^{(2)})^*_T)} \geq 1,
}
due to 
\eqstar{
\eta (L^{(1)})^*_T + (1-\eta) (L^{(2)})^*_T \geq m, \qquad \eta (L^{(1)})^*_T + (1-\eta) (L^{(2)})^*_T \geq (\eta L^{(1)} + (1-\eta) L^{(2)})^*_T).
}
Thus, we have shown that \eqref{eq:final step concavity} is indeed true. This gives us
\eqstar{
&\eta \E_t\left[U\Biggl(\frac{L^{(1)}_T}{(m \vee (L^{(1)})^*_T)}\Biggr) \right]+ (1-\eta) \E_t\left[U\Biggl(\frac{L^{(1)}_T}{(m \vee (L^{(1)})^*_T)}\Biggr) \right] \\
&\leq \E_t\left[U\Biggl(\frac{\eta L^{(1)}_T + (1-\eta) L^{(2)}_T}{(m \vee (\eta L^{(1)} + (1-\eta) L^{(2)})^*_T)}\Biggr) \right]\\
&\leq \cV(t, \eta l_1 + (1-\eta)l_2,m).
}
As, $\pi_1, \pi_2$ are arbitrary, we have have shown \eqref{eq:claim1}. This concludes the proof for concavity of $Q^{(0)}(t,\cdot).$
\end{proof}

\subsection{Proof of Proposition \ref{thm:risk tolerance}}\label{Prop2proof}
\begin{proof}
From the calculations performed in the proof of Proposition \ref{thm:zeroth const pde}, we know that 
\eqlnostar{eq:prop risk tol1}{
\del_{t \xi } Q^{(0)}  =  -\frac{\lambda_0^2}{2} Q^{(0)}_\xi\Bigl(-1+ \cR_\xi \Bigr).
}
Differentiating \eqref{eq:define risk tolerance} w.r.t.\ $t$ gives 
\eqlnostar{eq:prop risk tol2}{
\cR_t = -\frac{Q^{(0)}_{t \xi}}{Q^{(0)}_{\xi \xi}} + \frac{Q^{(0)}_{\xi}}{\bigl(Q^{(0)}_{\xi \xi}\bigr)^2} Q^{(0)}_{t \xi \xi}.
}
Differentiating \eqref{eq:prop risk tol1} w.r.t.\ $\xi,$ we get
\eqstar{
Q^{(0)}_{t \xi \xi }   =  -\frac{\lambda_0^2}{2} Q^{(0)}_{\xi\xi}\Bigl(-1+ \cR_\xi \Bigr) -\frac{\lambda_0^2}{2}Q^{(0)}_{\xi} \cR_{\xi \xi}.
}
Plugging back the above result and \eqref{eq:prop risk tol1} into \eqref{eq:prop risk tol2} gives the PDE for $\cR.$ The terminal condition at $t=T$ is straightforward from the terminal condition for $Q^{(0)}.$ At the boundary, $\xi = \alpha, Q^{(0)} = U(\alpha)$ which due to the continuity of $Q^{(0)}$ across the boundary gives that $Q_t^{(0)} = 0.$ Then, due to the continuity of derivatives w.r.t.\ space variables across the boundary, from \eqref{eq:hjbzeroth1} we get at $\xi = \alpha,$
\eqstar{
\frac{(Q_\xi^{(0)})^2}{Q_{\xi \xi}^{(0)}} = \cR Q_\xi^{(0)} = 0.
}
As $Q_\xi^{(0)} \big \vert_{\xi = \alpha} \neq 0,$ it gives that $\cR \big \vert_{\xi = \alpha} = 0.$

We know from our calculations in the proof of Lemma \ref{lem:zeroth order properties} and Section \ref{sec:optimal strategy} that the optimal strategy corresponding to the value function $Q^{(0)}$ is given as $\pi_0 := \text{constant} \times \cR.$ It is clear that as the portfolio wealth approaches to its maximum value, i.e.\ at $\xi = 1,$ the optimal strategy suggests to unwind the risky position, $\pi_0 \vert_{\xi = 1} = 0.$ This give us the right boundary condition for $\cR$ as $\cR\big \vert_{\xi = 1} = 0.$
\end{proof}

\subsection{Proof of Proposition \ref{thm:zeroth const pde}}\label{Prop1proof}
\begin{proof}
In the definition $Q^{(0)}(t,\xi) = q^{(0)}(t,z(t,\xi)),$ we differentiate w.r.t.\ $t$ on both sides to write 
\eqstar{
\del_t Q^{(0)}  &= \del_t q^{(0)} + q^{(0)}_z \frac{\del z}{\del t}\\
&= \del_t q^{(0)}  -  \Bigl( \frac{Q_{t\xi}^{(0)}}{Q_\xi^{(0)}} + \frac{\lambda^2_0}{2}\Bigr) q^{(0)}_z.
}
It is also straightforward to check from definition \eqref{eq:def diff operator} of differential operators $\bigl(\cD_k\bigr)_{k=1,2,\ldots}$ that 
\begin{equation}
\cD_1 Q^{(0)} = q^{(0)}_z, \quad \cD_2 Q^{(0)} =  q^{(0)}_{zz} - \cR_\xi  q^{(0)}_z.\label{derivs}
\end{equation}
Next, we observe that PDE \eqref{eq:hjbzeroth transform1} can also be written as $ Q^{(0)}_t = \frac{1}{2}\lambda_0^2 \cD_2 Q^{(0)}.$ Differentiating this w.r.t.\ $\xi,$ we get 
\eqstar{
\del_{t \xi } Q^{(0)} = \lambda_0^2\left( \frac{1}{2}  \cR^2 Q^{(0)}_{\xi\xi\xi} + \cR \cR_\xi Q^{(0)}_{\xi\xi}\right).
}
Further, from the definition of $\cR,$ we get $\cR Q^{(0)}_{\xi\xi} =  -Q^{(0)}_\xi$ which after differentiating w.r.t.\ $\xi$ gives
\eqstar{
\cR^2Q^{(0)}_{\xi\xi\xi} =  - \cR Q^{(0)}_{\xi\xi} \Bigl(1+ \cR_\xi\Bigr).
}
Thus, we have 
\eqstar{
\del_{t \xi } Q^{(0)}  &= \frac{\lambda_0^2}{2} \cR Q^{(0)}_{\xi\xi}\Bigl(-1+ \cR_\xi \Bigr)\\
&=  -\frac{\lambda_0^2}{2} Q^{(0)}_\xi\Bigl(-1+ \cR_\xi \Bigr).
}
Finally, we collect all the expressions for $\del_t Q^{(0)}, \cD_1 Q^{(0)}$ and $\cD_2 Q^{(0)}$ in terms of $q^{(0)}$ to write
\eqstar{
&\Bigl( \partial_t + \lambda_0^2 \cD_1 + \frac{\lambda_0^2}{2} \cD_2 \Bigr) Q^{(0)} \\
&= \del_t q^{(0)}  -  \Bigl( -\frac{\lambda_0^2}{2} \Bigl(-1+ \cR_\xi \Bigr) + \frac{\lambda^2_0}{2}\Bigr) q^{(0)}_z + \lambda_0^2 q^{(0)}_z + \frac{\lambda_0^2 }{2}\Bigl(q^{(0)}_{zz} - \cR_\xi  q^{(0)}_z\Bigr)\\
&= \Bigl(\frac{\partial}{\partial t} + \frac{1}{2} \lambda^2_0\frac{\partial}{\partial z^2}\Bigr) q^{(0)}
}
which gives us the desired PDE.

For the terminal boundary condition for $q^{(0)},$ it follows from the definition of $z(t,\xi)$ and terminal condition \eqref{eq:hjbzeroth_bound3} that 
\eqstar{
q^{(0)}(T,z) &= U\Bigl( \bigl(U^\prime\bigr)^{-1}\bigl( \ee^{-z}\bigr)\Bigr), \quad \psi(T) < z < \infty.
}
The left boundary condition in \eqref{eq:hjbzeroth_bound3} can also be easily transformed. Next, for the right boundary condition in \eqref{eq:hjbzeroth_bound3}, we first note that 
\eqstar{
q^{(0)}_z \times \partial_\xi z &= Q^{(0)}_\xi.
}
Now, as $Q^{(0)}_\xi = 0, \text{for } \xi = 1$, it holds only if in the above relation we have 
$$ \lim_{ z \to \infty} q^{(0)}_z(t,z)  = 0.$$
This completes the proof. 
\end{proof}

\subsection{Proof of Proposition \ref{prop:solution_firstorder}}\label{Prop4proof}
We first consider the PDE problem with a terminal condition
\eqlnostar{eq:source_equationfirstorder}{
\cH q + \cS = 0, \quad q(T,z,x,y) = 0,
}
where $\cH$ is a constant coefficient linear operator
\eqstar{
\cH := \frac{\del}{\del t} + \cA_0 + \cC_0.
}
We suppose that the source term $\cS$ is of the following special form
\eqlnostar{eq:sourceterm_firstorder}{
\cS(t,z,x,y) = \sum_{k,l,n} (T-t)^n (x-\bar{x})^k(y-\bar{y})^l v(t,z,x,y)
}
where the sum has a finite number of terms, and $v$ is a solution of the homogeneous equation $\cH v = 0$.

Further, define the commutator of operators $\cH$ and $(x-\bar{x})I$ ($I$ is the identity operator), $\cL_X = [\cH, (x-\bar{x}) I]$ as
\eqstar{
\cL_X v := \cH((x-\bar{x})v) - (x-\bar{x}) \cH v,
}
which from the definition of $\cA_0$ \eqref{eq:defA0} and $\cC_0$ \eqref{eq:defC0} gives
\eqlnostar{eq:operator X}{
\cL_X = b_0 I + \sigma^2_0 \frac{\del}{\del x} + \rho \sigma_0 \beta_0 \frac{\del}{\del y} + \sigma_0 \lambda_0 \frac{\del}{\del z}.
}
Similarly, define $\cL_Y = [\cH, (y - \bar{y}) I]$, which gives
\eqlnostar{eq:operator Y}{
\cL_Y = c_0 I + \beta^2_0 \frac{\del}{\del y} + \rho \sigma_0 \beta_0 \frac{\del}{\del x} + \rho \beta_0 \lambda_0 \frac{\del}{\del z}. 
}
Using $\cL_X$ and $\cL_Y,$ we also define
\eqstar{
\cM_X(s) := (x-\bar{x})I  + (s-t)\cL_X, \quad \cM_Y(s) := (y-\bar{y})I  + (s-t)\cL_Y.
}
Using these definitions, we first give the following result related to the homogeneous solution $v$, from \cite[Lemma 3.4]{lorig2015portfolio}. Here, we provide the proof for the sake of completeness. 
\begin{lemma}
\label{lemma:operator1}
For integers $k,l,$ we have,
\eqstar{
\cH \cM^k_X(s) \cM^l_Y(s) v = 0.
}
\end{lemma}
\begin{proof}
We proceed by induction. We first calculate 
\eqstar{
\cH \cM_X(s) v &= \cH (x- \bar{x})v + \cH (s-t) \cL_X v\\
&= \cL_X v + (x-\bar{x})\cH v - \cL_X v + (s-t) \cH \cL_X v\\
&= \cL_X \cH v = 0,
}
where we have used the definition of the commutator $\cL_X,$ the fact that $\cL_X$ and $\cH$ commute as they are constant coefficient operators and that $\cH v = 0.$ Thus, we can then iterate over integer $k$ to show $\cH \cM_X (\cM_X^{(k-1)}v) = 0 \, (\text{as } \cH \cM_X(s) v =0).$ Similarly, we can show that $\cH \cM_Y^l v = 0$ for integer $l.$ Finally, we have $\cH (\cM^k_X(s)$ $\cM^l_Y(s) v) = 0.$ 
\end{proof}

\begin{lemma}\label{sol_lemma}
The solution $q$ of equation \eqref{eq:source_equationfirstorder} with zero terminal condition is
\eqlnostar{eq:ansatz_correction1}{
q(t,z,x,y) = \sum_{k,l,n} \int ^T_t (T-s)^n \cM^k_X(s) \cM^l_Y(s)v(t,z,x,y)\, \dd s.
}
\end{lemma}

\begin{proof}
This can be shown by using the form of source term \eqref{eq:sourceterm_firstorder} and Lemma \ref{lemma:operator1}. Let us suppose that the source term consists of a monomial and is given as $\cS(t,z,x,y) = (T-t)^n (x-\bar{x})^k(y-\bar{y})^l v(t,z,x,y).$ In this case, from our claim, the solution should be given as 
$$q(t,z,x,y) = \int ^T_t (T-s)^n \cM^k_X(s) \cM^l_Y(s)v(t,z,x,y) \dd s.$$
We verify by computing 
\eqstar{
\cH q &= -(T-t)^n \cM^k_X(t) \cM^l_Y(t) v(t,z,x,y) + \int ^T_t (T-s)^n \cH \Bigl(\cM^k_X(s) \cM^l_Y(s)v(t,z,x,y)\Bigr) \dd s\\
&= -(T-t)^n (x-\bar{x})^k (y-\bar{y})^l v(t,z,x,y) \\
&= - \cS.
}
It is also easy to see that for the form of solution proposed in \eqref{eq:ansatz_correction1}, the terminal condition at $T$ is satisfied. The result follows from linearity of the PDE problem.
\end{proof}

Finally, we give the proof of Proposition \ref{prop:solution_firstorder}.

\begin{proof}
We first observe that, since $q^{(0)}$ solves $\cH q^{(0)} = 0$, then $q^{(0)}_z$ also solves the homogeneous equation, as the operator $\cH$ has constant coefficients. 
We set $v = q_z^{(0)}$.
From 
\eqref{eq:main source term}, the source term is 
\eqstar{
\cS(t,z,x,y) &=\Bigl( \bigl(\frac{1}{2} \lambda^2\bigr)_{1,0}(x-\bar{x}) + \bigl(\frac{1}{2} \lambda^2\bigr)_{0,1}(y-\bar{y}) \Bigr)v, 
}
and so from Lemma \ref{sol_lemma}, we obtain the solution
\begin{align}
q^{(1)}(t,z,x,y) =& \Bigl[\bigl(\frac{1}{2} \lambda^2\bigr)_{1,0}\bigl((T-t)(x-\bar{x}) + \frac{1}{2}(T-t)^2 \cL_X\bigr) \\
&+ \bigl(\frac{1}{2} \lambda^2\bigr)_{0,1}\bigl((T-t)(y-\bar{y}) + \frac{1}{2}(T-t)^2 \cL_Y \bigr) \Bigr] q^{(0)}_z(t,z). \label{q1exp}
\end{align}
From the expansion for $\lambda(y),$ we get 
\eqstar{
\bigl(\frac{1}{2}\lambda^2\bigr)_{1,0} =  \lambda_0  \lambda_{1,0}, \quad \bigl(\frac{1}{2}\lambda^2\bigr)_{0,1} =  \lambda_0  \lambda_{0,1}.
}
Putting back the expression of $\cL_X$ and $\cL_Y$ from \eqref{eq:operator X} and \eqref{eq:operator Y} 
into \eqref{q1exp}, we get the expression in \eqref{eq:first correction term}. The terminal condition at $t=T$ is clearly satisfied. 

It remains to check the boundary conditions for $q^{(1)}$. We show that the boundary conditions for $Q^{(1)}$, corresponding to the original variables $(t,\xi)$, are satisfied. Using \eqref{q1def} and \eqref{derivs}, we obtain \eqref{eq:first order final}. Now, due to the zero boundary condition 
at $\xi = \alpha$ for the risk-tolerance function $\cR,$ we get from \eqref{eq:first order final} that 
$Q^{(1)}(t,\alpha,x,y) =  0$,
which means that the left boundary condition in \eqref{eq:hjbbound11} is satisfied. Consequently, the left boundary condition in \eqref{eq:transformed_boundary1} is satisfied for $q^{(1)}.$ 

Next, we calculate
\eqlnostar{eq:intercorrection3}{
Q^{(1)}_{\xi}(t,\xi,x,y) &= (T-t)\lambda_0 A(t,x,y)\Bigl(\cR_{\xi} Q^{(0)}_{\xi} + \cR Q^{(0)}_{\xi\xi} \Bigr) + \frac{1}{2}(T-t)^2 \lambda_0 B \Bigl( -2  \bigl[\cR_{\xi} Q^{(0)}_{\xi} + \cR Q^{(0)}_{\xi \xi} \bigr] \nonumber\\
&+ 3\cR^2 \partial^3_\xi Q^{(0)} +\cR^3 \partial^4_\xi Q^{(0)} \Bigr).
}
From our Assumption \ref{reg-assump} on the boundedness of $\del^{k}_\xi Q^{(0)}(t,1)$ for $k\leq5$, 
we have 
 \[ \lim_{\xi \to 1}\cR^{k} \del^{(k+1)}_\xi Q^{(0)} = 0, \qquad k = 1,2,3. \]
Then, we can use the boundary condition of $Q^{(0)}_{\xi}$ and $\cR$ at $\xi = 1$ 
to conclude from \eqref{eq:intercorrection3} that 
$$Q^{(1)}_{\xi}(t,\xi,x,y) \Big \vert_{\xi = 1} = 0,$$
which means that the right boundary condition in \eqref{eq:hjbbound11} is satisfied. This implies that the right boundary condition in \eqref{eq:transformed_boundary1} is satisfied for $q^{(1)}$.
\end{proof}

\bibliographystyle{abbrvnat}
\bibliography{utilitymaxasymptotics}

\begin{thebibliography}{21}
\providecommand{\natexlab}[1]{#1}
\providecommand{\url}[1]{\texttt{#1}}
\expandafter\ifx\csname urlstyle\endcsname\relax
  \providecommand{\doi}[1]{doi: #1}\else
  \providecommand{\doi}{doi: \begingroup \urlstyle{rm}\Url}\fi

\bibitem[Chacko and Viceira(2005)]{chacko2005dynamic}
G.~Chacko and L.~M. Viceira.
\newblock Dynamic consumption and portfolio choice with stochastic volatility
  in incomplete markets.
\newblock \emph{Review of Financial Studies}, 18\penalty0 (4):\penalty0
  1369--1402, 2005.

\bibitem[Chekhlov et~al.(2005)Chekhlov, Uryasev, and
  Zabarankin]{chekhlov2005drawdown}
A.~Chekhlov, S.~Uryasev, and M.~Zabarankin.
\newblock Drawdown measure in portfolio optimization.
\newblock \emph{International Journal of Theoretical and Applied Finance},
  8\penalty0 (01):\penalty0 13--58, 2005.

\bibitem[Chen et~al.(2015)Chen, Landriault, Li, and Li]{chen2015minimizing}
X.~Chen, D.~Landriault, B.~Li, and D.~Li.
\newblock On minimizing drawdown risks of lifetime investments.
\newblock \emph{Insurance: Mathematics and Economics}, 65:\penalty0 46--54,
  2015.

\bibitem[Cherny and Ob{\l}{\'o}j(2013)]{cherny2013portfolio}
V.~Cherny and J.~Ob{\l}{\'o}j.
\newblock Portfolio optimisation under non-linear drawdown constraints in a
  semimartingale financial model.
\newblock \emph{Finance and Stochastics}, 17\penalty0 (4):\penalty0 771--800,
  2013.

\bibitem[Cvitanic and Karatzas(1995)]{cvitanic1995portfolio}
J.~Cvitanic and I.~Karatzas.
\newblock On portfolio optimization under "drawdown" constraints.
\newblock \emph{IMA Volumes in Mathematics and its Applications}, 65:\penalty0
  35--35, 1995.

\bibitem[Elie(2008)]{elie2008finite}
R.~Elie.
\newblock Finite time merton strategy under drawdown constraint: a viscosity
  solution approach.
\newblock \emph{Applied Mathematics and Optimization}, 58\penalty0
  (3):\penalty0 411--431, 2008.

\bibitem[Elie and Touzi(2008)]{elie2008optimal}
R.~Elie and N.~Touzi.
\newblock Optimal lifetime consumption and investment under a drawdown
  constraint.
\newblock \emph{Finance and Stochastics}, 12\penalty0 (3):\penalty0 299--330,
  2008.

\bibitem[Fouque et~al.(2015)Fouque, Sircar, and
  Zariphopoulou]{fouque2015portfolio}
J.-P. Fouque, R.~Sircar, and T.~Zariphopoulou.
\newblock Portfolio optimization and stochastic volatility asymptotics.
\newblock \emph{Mathematical Finance}, 2015.

\bibitem[Grossman and Zhou(1993)]{grossman1993optimal}
S.~J. Grossman and Z.~Zhou.
\newblock Optimal investment strategies for controlling drawdowns.
\newblock \emph{Mathematical Finance}, 3\penalty0 (3):\penalty0 241--276, 1993.

\bibitem[K\"{a}llblad and Zariphopoulou(2014)]{kallblad2014qualitative}
S.~K\"{a}llblad and T.~Zariphopoulou.
\newblock Qualitative analysis of optimal investment strategies in log-normal
  markets.
\newblock \emph{Available at SSRN 2373587}, 2014.

\bibitem[Karatzas et~al.(1987)Karatzas, Lehoczky, and
  Shreve]{karatzas1987optimal}
I.~Karatzas, J.~P. Lehoczky, and S.~E. Shreve.
\newblock Optimal portfolio and consumption decisions for a "small investor" on
  a finite horizon.
\newblock \emph{SIAM Journal on Control and Optimization}, 25\penalty0
  (6):\penalty0 1557--1586, 1987.

\bibitem[Lorig and Sircar(2016)]{lorig2015portfolio}
M.~Lorig and R.~Sircar.
\newblock Portfolio optimization under local-stochastic volatility: Coefficient
  taylor series approximations \& implied sharpe ratio.
\newblock \emph{SIAM J. Financial Mathematics}, 7:\penalty0 418--447, 2016.

\bibitem[Lorig et~al.(2015)Lorig, Pagliarani, and Pascucci]{lorig2015explicit}
M.~Lorig, S.~Pagliarani, and A.~Pascucci.
\newblock Explicit implied volatilities for multifactor local-stochastic
  volatility models.
\newblock \emph{Mathematical Finance}, 2015.
\newblock ISSN 1467-9965.

\bibitem[Magdon-Ismail and Atiya(2004)]{magdon2004maximum}
M.~Magdon-Ismail and A.~F. Atiya.
\newblock Maximum drawdown.
\newblock \emph{Risk Magazine}, 17\penalty0 (10):\penalty0 99--102, 2004.

\bibitem[Merton(1969)]{merton1969lifetime}
R.~C. Merton.
\newblock Lifetime portfolio selection under uncertainty: The continuous-time
  case.
\newblock \emph{The Review of Economics and Statistics}, pages 247--257, 1969.

\bibitem[Merton(1971)]{merton1971optimum}
R.~C. Merton.
\newblock Optimum consumption and portfolio rules in a continuous-time model.
\newblock \emph{Journal of Economic Theory}, 3\penalty0 (4):\penalty0 373--413,
  1971.

\bibitem[Pham(2009)]{pham2009continuous}
H.~Pham.
\newblock \emph{Continuous-time stochastic control and optimization with
  financial applications}, volume~61.
\newblock Springer Science \& Business Media, 2009.

\bibitem[Roche(2006)]{roche2006optimal}
H.~Roche.
\newblock Optimal consumption and investment strategies under wealth
  ratcheting.
\newblock \emph{Preprint}, 2006.
\newblock URL \url{http://ciep.itam.mx/~hroche/Research/MDCRESFinal.pdf}.

\bibitem[Rogers(2013)]{rogers2013optimal}
L.~C. Rogers.
\newblock \emph{Optimal investment}.
\newblock Springer, 2013.

\bibitem[Samuelson(1969)]{samuelson1969lifetime}
P.~A. Samuelson.
\newblock Lifetime portfolio selection by dynamic stochastic programming.
\newblock \emph{The Review of Economics and Statistics}, pages 239--246, 1969.

\bibitem[Sekine(2013)]{sekine2013long}
J.~Sekine.
\newblock Long-term optimal investment with a generalized drawdown constraint.
\newblock \emph{SIAM Journal on Financial Mathematics}, 4\penalty0
  (1):\penalty0 452--473, 2013.

\end{thebibliography}

\end{document}